\documentclass[twoside]{article}
\usepackage[accepted]{aistats2017}

\usepackage{tikz}
\usetikzlibrary{fit,positioning,shapes.geometric}

\tikzset{
  stochastic/.style={
    circle,
    minimum size=7mm, 
    thick, 
    draw=black!80, 
    node distance=5mm
  },
  deterministic/.style={
    regular polygon,
    regular polygon sides=4, 
    inner sep=0, 
    outer sep=0,
    minimum size=9mm, 
    thick, 
    draw=black!80, 
    node distance=5mm
  },
  observed/.style={
    fill=lightgray
  },
  directed/.style={
    -latex, 
    thick
  },
  undirected/.style={
    -, 
    thick
  }
}

\usepackage{hyperref}       
\usepackage{url}            
\usepackage{booktabs}       
\usepackage{amsfonts}       
\usepackage{nicefrac}       
\usepackage{microtype}      
\usepackage{enumitem}
\usepackage{nameref}
\usepackage{zref-xr}
\zxrsetup{toltxlabel}
\usepackage[square,sort,comma,numbers]{natbib}
\usepackage{amsmath}
\usepackage{amsfonts}
\usepackage{amssymb}
\usepackage{amsthm}
\usepackage{mathrsfs} 
\usepackage{accents} 
\usepackage{stix}

\usepackage{float}
\usepackage{graphicx}
\usepackage[subrefformat=parens]{subcaption}

\usepackage{multicol}
\usepackage[noend]{algpseudocode}
\usepackage{algorithm}
\usepackage{setspace}

\newcommand{\vct}[1]{\boldsymbol{#1}}
\newcommand{\mtx}[1]{\boldsymbol{\expandafter\MakeUppercase\expandafter{#1}}}
\newcommand{\set}[1]{\mathcal{#1}}
\newcommand{\fset}[1]{\lbr #1 \rbr}
\newcommand{\reals}{\mathbb{R}}

\newcommand{\lpa}{\left(}
\newcommand{\rpa}{\right)}
\newcommand{\lbr}{\left\lbrace}
\newcommand{\rbr}{\right\rbrace}
\newcommand{\lsb}{\left[}
\newcommand{\rsb}{\right]}

\newcommand{\pden}[2]{\mathbb{p}_{\,\scriptstyle #1\,}{\lsb #2 \rsb}}

\newcommand{\rvar}[1]{\mathrm{#1}}
\newcommand{\rvct}[1]{\mathbf{#1}}
\newcommand{\nrm}[1]{\mathcal{N}\lpa #1 \rpa}

\newcommand{\expc}[2]{\mathbb{E}_{#1}\lsb #2 \rsb}

\newcommand{\dr}{\mathrm{d}}
\newcommand{\td}[2]{\frac{\dr #1}{\dr #2}}
\newcommand{\pd}[2]{\frac{\partial #1}{\partial #2}}
\newcommand{\pdd}[3]{\frac{\partial^2 #1}{\partial #2 \partial #3}}
\newcommand{\tr}{^\mathrm{T}}
\newcommand{\sml}[1]{#1}
\newcommand{\obs}[1]{\bar{#1}}
\newcommand{\lebm}[2]{\lambda^{#1}\lbr \dr #2 \rbr}
\newcommand{\haum}[2]{\mathcal{H}^{#1}\lbr \dr #2 \rbr}

\DeclareMathOperator*{\chol}{chol}
\DeclareMathOperator*{\gvn}{|}
\DeclareMathOperator*{\trace}{Trace}

\newcommand{\abbrdef}[2]{\textit{#1}~(\textrm{#2})}
\newcommand{\abbrref}[1]{\textrm{#1}}

\newcommand{\func}[1]{\mathit{#1}}
\newcommand{\vctfunc}[1]{{\boldsymbol{\expandafter\MakeUppercase\expandafter{#1}}}}
\newcommand{\genfunc}{\vctfunc{g}}

\newcommand{\uconfunc}{\vctfunc{c}}

\newcommand{\inlinetr}{\phantom{}^\textsc{t}}

\makeatletter 
  \g@addto@macro\@uclclist{%
    \eth\Eth
    \thorn\Thorn
    \alpha\Alpha
    \beta\Beta
    \gamma\Gamma
    \delta\Delta
    \epsilon\Epsilon
    \varepsilon\Varepsilon
    \zeta\Zeta
    \eta\Eta
    \theta\Theta
    \vartheta\Vartheta
    \iota\Iota
    \kappa\Kappa
    \lambda\Lambda
    \mu\Mu
    \nu\Nu
    \xi\Xi
    \omicron\Omicron
    \pi\Pi
    \varpi\Varpi
    \rho\Rho
    \varrho\Varrho
    \sigma\Sigma
    \varsigma\Varsigma
    \tau\Tau
    \upsilon\Upsilon
    \phi\Phi
    \varphi\Varphi
    \chi\Chi
    \psi\Psi
    \omega\Omega
  }

\algrenewcommand\alglinenumber[1]{\tiny #1:}

\newtheoremstyle{thmstyle}
  {4pt} 
  {4pt} 
  {\it} 
  {} 
  {\bfseries} 
  {.} 
  {.5em} 
  {} 

\theoremstyle{thmstyle} \newtheorem*{theorem}{Theorem}

\newtheorem*{corollary}{Corollary}
\zexternaldocument*{supplementary}

\begin{document}

\twocolumn[
  \aistatstitle{Asymptotically exact inference in differentiable generative models}
  \aistatsauthor{ Matthew M. Graham \And Amos J. Storkey }
  \aistatsaddress{University of Edinburgh \And University of Edinburgh}
]

\begin{abstract}  
Many generative models can be expressed as a differentiable function of random inputs drawn from some simple probability density. This framework includes both deep generative architectures such as Variational Autoencoders and a large class of procedurally defined simulator models. We present a method for performing efficient MCMC inference in such models when conditioning on observations of the model output. For some models this offers an asymptotically exact inference method where Approximate Bayesian Computation might otherwise be employed. We use the intuition that inference corresponds to integrating a density across the manifold corresponding to the set of inputs consistent with the observed outputs. This motivates the use of a constrained variant of Hamiltonian Monte Carlo which leverages the smooth geometry of the manifold to coherently move between inputs exactly consistent with observations. We validate the method by performing inference tasks in a diverse set of models.
\end{abstract}

\section{Introduction}\label{sec:intro}

Developments in generative modelling with deep architectures such as \abbrdef{Variational Auto-Encoders}{VAEs} \citep{kingma2013auto,rezende2014stochastic} and \abbrdef{Generative Adversarial Nets}{GANs} \citep{goodfellow2014generative} have made it possible to learn probabilistic models of increasingly complex, high-dimensional data-sets. The generator of these models can be formulated as a differentiable\footnote{Differentiable here means that the Jacobian $\pd{\genfunc}{\vct{u}}$ is defined a.e.} function $\genfunc : \set{U} \to \set{X}$ which maps random vector inputs $\rvct{u} \in \set{U} = \reals^M$ drawn from a probability distribution with known density $\pden{\rvct{u}}{\vct{u}} = \rho(\vct{u})$ to an implicitly defined distribution over random vector outputs $\rvct{x} = \genfunc(\rvct{u}) \in \set{X} = \reals^N$. We will refer here to a model of this form as a \abbrdef{differentiable generative model}{DGM}.

As well as parametric models learnt from data, the \abbrref{DGM} framework also encapsulates a wide class of simulator models where $\genfunc$ is defined procedurally, e.g. numerical integration of a system of stochastic differential equations. Here the random inputs $\rvct{u}$ are the series of draws from a random number generator during the simulator's execution. The operations used in many simulations are differentiable and automatic differentiation provides a computationally efficient framework for calculating the exact derivatives of a differentiable simulator's outputs with respect to its random inputs given just the code used to define the model. 

Often we will be interested in using a generative model to make inferences about the modelled variables given observations related to outputs of the model. For example given a generative model of images we may wish to infer plausible in-paintings of an image region given knowledge of the surrounding pixel values. Similarly given a simulator of a physical process and generator of parameters of the process which we believe are reasonable a priori, we may wish to infer our posterior beliefs about the parameters under the model given observations of the physical process.

Concretely, we consider the generated output $\rvct{x}$ as being partitioned into observed variables $\rvct{y} \in \set{Y} = \reals^{N_y}$ and latent variables $\rvct{z}\in\set{Z}=\reals^{N_z}$ to be inferred, i.e. $\rvct{x} = [\rvct{y};\,\rvct{z}]$ and $\set{X}=\set{Y}\times\set{Z}$. Likewise we partition the generator such that $\rvct{y} = \genfunc_{y}(\rvct{u})$ and $\rvct{z} = \genfunc_{z}(\rvct{u})$. Inference is then the task of computing expectations of $\rvct{z}$ conditioned on an observed $\rvct{y}$. 

A common special case is where $\rvct{z}$ is generated from a subset of the random inputs $\rvct{u}_{z}$, and $\rvct{y}$ is then generated as a function of $\rvct{z}$ and the remaining random inputs $\rvct{u}_{y}$. For example when $\rvct{z}$ is sampled from a \emph{prior} and $\rvct{y}$ sampled from a \emph{likelihood} given $\rvct{z}$. We will refer to this specialism as a \emph{directed model} in reference to the fact that it has a natural interpretation as a directed graphical model as illustrated in figure \ref{fig:model-graphs}.

\begin{figure}
\centering
\begin{subfigure}[t]{.3\linewidth}
\centering
\begin{tikzpicture}
  \node at (0, 0.5) [stochastic] (u) {$\rvct{u}$};
  \node at (1, 1) [deterministic] (z) {$\rvct{z}$};
  \node at (1, 0) [deterministic, observed] (y) {$\rvct{y}$};
  \path (u) edge [directed] (y)
        (u) edge [directed] (z);
\end{tikzpicture}
\caption{}\label{sfig:undirected-sgc}
\end{subfigure}%
\begin{subfigure}[t]{.3\linewidth}
\centering
\begin{tikzpicture}
  \node at (0, 1) [stochastic] (u_z) {$\rvct{u}_{z}$};
  \node at (0, 0) [stochastic] (u_y) {$\rvct{u}_{y}$};
  \node at (1, 1) [deterministic] (z) {$\rvct{z}$};
  \node at (1, 0) [deterministic, observed] (y) {$\rvct{y}$};
  \path (u_y) edge [directed] (y)
        (z) edge [directed] (y)
        (u_z) edge [directed] (z);
\end{tikzpicture}
\caption{}\label{sfig:directed-sgc}
\end{subfigure}%
\begin{subfigure}[t]{.2\linewidth}
\centering
\begin{tikzpicture}
  \node at (0, 1) [stochastic] (z_gm) {$\rvct{z}$};
  \node at (0, 0) [stochastic, observed] (y_gm) {$\rvct{y}$};
  \path (z_gm) edge [undirected] (y_gm);
\end{tikzpicture}
\caption{}\label{sfig:undirected-gm}
\end{subfigure}%
\begin{subfigure}[t]{.2\linewidth}
\centering
\begin{tikzpicture}
  \node at (0, 1) [stochastic] (z_gm) {$\rvct{z}$};
  \node at (0, 0) [stochastic, observed] (y_gm) {$\rvct{y}$};
  \path (z_gm) edge [directed] (y_gm);
\end{tikzpicture}
\caption{}\label{sfig:directed-gm}
\end{subfigure}%
\caption{\subref{sfig:undirected-sgc} and \subref{sfig:directed-sgc} \abbrdef{Stochastic computation graphs}{SGC} \citep{schulman2015gradient} visualising models considered in this paper. Circular nodes represent stochastic variables sampled from a conditional distribution on their parent nodes or marginal for root nodes. Square nodes represent variables that are deterministic functions of their parents. Shaded nodes are observed. \subref{sfig:undirected-sgc} \abbrref{SGC} of general case in which $\rvct{y}$ and $\rvct{z}$ are jointly generated from $\rvct{u}$. \subref{sfig:directed-sgc} \abbrref{SGC} of the directed case in which we first generate $\rvct{z}$ from $\rvct{u}_z$ then generate $\rvct{y}$ from $\rvct{z}$ and $\rvct{u}_y$. \subref{sfig:undirected-gm} Undirected graphical model corresponding to \subref{sfig:undirected-sgc} when integrating out $\rvct{u}$. \subref{sfig:directed-gm} Directed graphical model which is a natural representation of \subref{sfig:directed-sgc} when integrating out $\rvct{u}_y$ and $\rvct{u}_z$.}
\label{fig:model-graphs}
\end{figure}
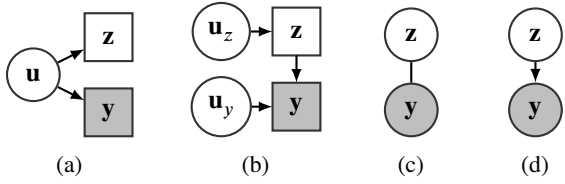

For many \abbrref{DGMs} we cannot explicitly compute the joint density on the generator outputs $\pden{\rvct{x}}{\vct{x}} = \pden{\rvct{y},\,\rvct{z}}{\vct{y},\,\vct{z}}$ and so conditional density $\pden{\rvct{z}\gvn\rvct{y}}{\vct{z}\gvn\vct{y}}$ we wish to compute expectations with respect to. In directed models this usually corresponds to not being able to evaluate the likelihood $\pden{\rvct{y}\gvn\rvct{z}}{\vct{y} \gvn \vct{z}}$. The lack of a closed form density impedes the direct use of approximate inference methods such as variational inference and \abbrdef{Markov chain Monte Carlo}{MCMC}.

\section{Approximate Bayesian Computation}\label{sec:abc}

A lack of explicit likelihoods is the motivation for \emph{likelihood-free} inference methods such as \abbrdef{Approximate Bayesian Computation}{ABC} \citep{beaumont2002approximate,marin2012approximate}. In \abbrref{ABC} the simulated observed outputs $\sml{\rvct{y}}$ are decoupled from the observed data $\obs{\rvct{y}}$ by a noise model or \emph{kernel} $\pden{\obs{\rvct{y}}\gvn\sml{\rvct{y}}}{\obs{\vct{y}}\gvn\sml{\vct{y}}} = k_\epsilon(\obs{\vct{y}};\,\sml{\vct{y}})$ with tolerance parameter $\epsilon$, e.g. $k_\epsilon(\obs{\vct{y}};\,\sml{\vct{y}}) \propto \mathbb{I}\lsb \left|\obs{\vct{y}} - \sml{\vct{y}}\right| < \epsilon \rsb / \epsilon^{N_y}$ (uniform ball kernel) or $k_\epsilon(\obs{\vct{y}};\,\sml{\vct{y}}) = \nrm{\obs{\vct{y}}\gvn\sml{\vct{y}},\,\epsilon^2\mtx{I}}$ (Gaussian kernel). The \emph{\abbrref{ABC} likelihood} is then defined as
\begin{equation}\label{eq:abc_likelihood}
  \pden{\obs{\rvct{y}} \gvn \rvct{z}}{\obs{\vct{y}} \gvn \vct{z}; \epsilon} =
  \int_{\set{Y}}
    k_\epsilon(\obs{\vct{y}};\,\sml{\vct{y}})\,
    \pden{\sml{\rvct{y}} \gvn \rvct{z}}{\sml{\vct{y}} \gvn \vct{z}}
  \,\dr \sml{\vct{y}}.
\end{equation}
The \emph{\abbrref{ABC} posterior} is likewise defined as
\begin{equation}\label{eq:abc-posterior}
  \pden{\rvct{z}\gvn\obs{\rvct{y}}}{\vct{z}\gvn\obs{\vct{y}};\,\epsilon} =
  \frac{
  \int_{\set{Y}}
    k_\epsilon(\obs{\vct{y}};\,\sml{\vct{y}})\,
    \pden{\sml{\rvct{y}} \gvn \rvct{z}}{\sml{\vct{y}} \gvn \vct{z}}\,
    \pden{\rvct{z}}{\vct{z}}
  \,\dr \sml{\vct{y}}
  }
  {
  \int_{\set{Y}}
    k_\epsilon(\obs{\vct{y}};\,\sml{\vct{y}})\,
    \pden{\sml{\rvct{y}}}{\sml{\vct{y}}}
  \,\dr \sml{\vct{y}}
  }  
\end{equation}
and we can express expectations of functions of the latent variables $\rvct{z}$ with respect to this approximate posterior
\begin{align}\label{eq:abc-expectations}
  &\expc{}{f(\rvct{z}) \gvn \obs{\rvct{y}}=\obs{\vct{y}};\,\epsilon} =
  \int_{\set{Z}} 
    f(\vct{z}) \,
    \pden{\rvct{z}\gvn\obs{\rvct{y}}}{\vct{z}\gvn\obs{\vct{y}};\,\epsilon}
  \,\dr\vct{z}
  \\
  &=
  \frac{
  \int_\set{Z} \int_{\set{Y}}
      f(\vct{z})\,
      k_\epsilon(\obs{\vct{y}};\,\sml{\vct{y}})\,
      \pden{\sml{\rvct{y}} \gvn \rvct{z}}{\sml{\vct{y}} \gvn \vct{z}}\,
      \pden{\rvct{z}}{\vct{z}}
  \,\dr \sml{\vct{y}}
  \,\dr{\vct{z}}
  }
  {
  \int_\set{Z} \int_{\set{Y}}
      k_\epsilon(\obs{\vct{y}};\,\sml{\vct{y}})\,
      \pden{\sml{\rvct{y}} \gvn \rvct{z}}{\sml{\vct{y}} \gvn \vct{z}}\,
      \pden{\rvct{z}}{\vct{z}}
  \,\dr \sml{\vct{y}}
  \,\dr{\vct{z}}
  }.
  \nonumber
\end{align}
Various Monte Carlo inference schemes can be used to estimate this expectation. The simplest is to generate a set of independent samples $\lbr \vct{y}^{(s)},\,\vct{z}^{(s)} \rbr_{s=1}^S$ from $\pden{\rvct{y},\,\rvct{z}}{\vct{y},\,\vct{z}}$ \footnote{\abbrref{ABC} is usually applied to directed models hence this is usually considered as generating $\rvct{z}$ from a prior then simulating $\rvct{y}$ given $\rvct{z}$ however more generally we can just sample from the joint.} and use an importance sampling estimator
\begin{equation}
  \expc{}{f(\rvct{z}) \gvn \obs{\rvct{y}}=\obs{\vct{y}};\,\epsilon} \approx
  \frac{
  \sum_{s=1}^S \lbr 
    f\lpa\vct{z}^{(s)}\rpa \, 
    k_\epsilon\lpa\obs{\vct{y}};\,\sml{\vct{y}}^{(s)}\rpa 
  \rbr
  }
  {
  \sum_{s=1}^S \lbr 
    k_\epsilon\lpa\obs{\vct{y}};\,\sml{\vct{y}}^{(s)}\rpa 
  \rbr
  }.
\end{equation}
In the case of a uniform ball kernel this corresponds to the standard ABC reject algorithm, with expectations being estimated as averages over the latent variable samples for which the corresponding simulated outputs are within a (Euclidean) distance of $\epsilon$ from the observations.

Alternatively a \abbrref{MCMC} scheme can be used to estimate the \abbrref{ABC} posterior expectation in  \eqref{eq:abc-expectations} \citep{marjoram2003markov}. A Markov chain is constructed with stationary distribution
\begin{equation}\label{eq:abc-mcmc-target}
  \pden{\rvct{y},\,\rvct{z}\gvn\obs{\rvct{y}}}{\vct{y},\,\vct{z}\gvn\obs{\vct{y}}} \propto
  k_\epsilon(\obs{\vct{y}};\,\sml{\vct{y}})\,
  \pden{\sml{\rvct{y}} \gvn \rvct{z}}{\sml{\vct{y}} \gvn \vct{z}}\,
  \pden{\rvct{z}}{\vct{z}}
\end{equation}
by proposing a new state $(\vct{y}',\,\vct{z}')$ given the current state $(\vct{y},\,\vct{z})$ by sampling $\vct{z}' \sim q(\cdot\gvn\vct{z})$ from some perturbative proposal distribution $q$ on $\set{Z}$ and then generating a new $\vct{y}' \sim \pden{\sml{\rvct{y}} \gvn \rvct{z}}{\cdot \gvn \vct{z}}$. This is then accepted with probability
\begin{equation}\label{eq:abc-mcmc-accept}
  a(\vct{y}',\vct{z}' \gvn \vct{y},\vct{z}) =
  \min\lsb 1,\,
    \frac
    {k_\epsilon(\obs{\vct{y}};\,\sml{\vct{y}}')\,q(\vct{z}\gvn\vct{z}')\,\pden{\rvct{z}}{\vct{z}'}}
    {k_\epsilon(\obs{\vct{y}};\,\sml{\vct{y}})\,q(\vct{z}'\gvn\vct{z})\,\pden{\rvct{z}}{\vct{z}}}
  \rsb,
\end{equation}
the overall transition operator leaving \eqref{eq:abc-mcmc-target} stationary. The samples of the chain state can  then be used to compute consistent estimators of \eqref{eq:abc-expectations}. This scheme relies on being able to evaluate $\pden{\rvct{z}}{\vct{z}}$ and so is only applicable to directed models where the generator $\genfunc_z$ is of a tractable form such that $\pden{\rvct{z}}{\vct{z}}$ is known.

In the limit of $\epsilon \to 0$ valid \abbrref{ABC} kernels $k_\epsilon(\obs{\vct{y}};\,\sml{\vct{y}})$ collapse to a Dirac delta distribution ${\delta}\lpa \obs{\vct{y}} - \sml{\vct{y}}\rpa$ and so the \abbrref{ABC} posterior in \eqref{eq:abc-posterior} converges to the posterior $\pden{\rvct{z}\gvn\rvct{y}}{\vct{z}\gvn\obs{\vct{y}}}$. Similarly for the \abbrref{ABC} posterior expectation in \eqref{eq:abc-expectations} we have
\begin{align}\label{eq:abc-expectation-limit}
  &\lim_{\epsilon\to0} \lbr \expc{}{f(\rvct{z}) \gvn \obs{\rvct{y}}=\obs{\vct{y}};\,\epsilon} \rbr\nonumber
  \\
  &=
  \frac{
  \int_\set{Z} \int_{\set{Y}}
      f(\vct{z})\,
      {\delta}\lpa \obs{\vct{y}} - \sml{\vct{y}}\rpa\,
      \pden{\sml{\rvct{y}} \gvn \rvct{z}}{\sml{\vct{y}} \gvn \vct{z}}\,
      \pden{\rvct{z}}{\vct{z}}
  \,\dr \sml{\vct{y}}
  \,\dr{\vct{z}}
  }
  {
  \int_\set{Z} \int_{\set{Y}}
      {\delta}\lpa \obs{\vct{y}} - \sml{\vct{y}}\rpa\,
      \pden{\sml{\rvct{y}} \gvn \rvct{z}}{\sml{\vct{y}} \gvn \vct{z}}\,
      \pden{\rvct{z}}{\vct{z}}
  \,\dr \sml{\vct{y}}
  \,\dr{\vct{z}}
  }
  \\
  &= 
  \frac{
  \int_\set{Z}
      f(\vct{z})\,
      \pden{\sml{\rvct{y}} \gvn \rvct{z}}{\obs{\vct{y}} \gvn \vct{z}}\,
      \pden{\rvct{z}}{\vct{z}}
  \,\dr{\vct{z}}
  }
  {
  \int_\set{Z}
      \pden{\sml{\rvct{y}} \gvn \rvct{z}}{\obs{\vct{y}} \gvn \vct{z}}\,
      \pden{\rvct{z}}{\vct{z}}
  \,\dr{\vct{z}}
  }
  = 
  \expc{}{f(\rvct{z}) \gvn \rvct{y}=\obs{\vct{y}}}
  \nonumber.
\end{align}
For both the \abbrref{ABC} reject and \abbrref{ABC} \abbrref{MCMC} schemes just described, simulated observations $\vct{y}$ are independently generated from the conditional $\pden{\sml{\rvct{y}} \gvn \rvct{z}}{\vct{y} \gvn \vct{z}}$ given the current latent variables $\vct{z}$. The observed values $\obs{\vct{y}}$ are a zero-measure set in $\set{Y}$ under non-degenerate $\pden{\sml{\rvct{y}} \gvn \rvct{z}}{\vct{y} \gvn \vct{z}}$ and so as $\epsilon \to 0$ the probability of accepting a sample / proposed move becomes zero. Applying \abbrref{ABC} with a non-zero $\epsilon$ therefore can be seen as a practically motivated relaxation of the constraint that true and simulated data exactly match, and hence the `approximate' in \emph{Approximate Bayesian Computation}.

Alternatively the kernel $k_\epsilon$ can be given a modelling interpretation as representing uncertainty introduced by noise in the observations or mismatch between the unknown generative process by which the observed data was produced and the generative model \citep{ratmann2009model,wilkinson2013approximate}. In practice however the kernel and choice of $\epsilon$ seems more often to be motivated on computational efficiency grounds \citep{robert2010model}.

\section{Inference in the input space}\label{sec:cond-inf}

Using the \emph{Law of the Unconscious Statistician} we can reparametrise \eqref{eq:abc-expectations} in terms of the random inputs $\rvct{u}$ to the generative model
\begin{equation*}\label{eq:abc-expectation-input-space}
  \expc{}{f(\rvct{z}) \gvn \obs{\rvct{y}}=\obs{\vct{y}};\,\epsilon} \propto
  \int_{\set{U}} 
    f\,\circ\,\genfunc_z(\vct{u}) \,
    k_\epsilon{\lpa\obs{\vct{y}};\,\genfunc_y(\vct{u})\rpa}\,
    \rho(\vct{u})
  \,\dr\vct{u}.
\end{equation*}
This immediately indicates we can estimate \abbrref{ABC} expectations by applying any \abbrref{MCMC} method of choice in the input space to construct a chain with stationary density 
\begin{equation}\label{eq:abc-density-input-space}
  \pi_\epsilon(\vct{u}) \propto 
  k_\epsilon{\lpa\obs{\vct{y}};\,\genfunc_y(\vct{u})\rpa}\,\rho(\vct{u}).
\end{equation}

We can also however again consider the limit of $\epsilon \to 0$ as derived in the output space in \eqref{eq:abc-expectation-limit}. We have that the kernel term $k_\epsilon\lpa\obs{\vct{y}};\,\genfunc_y(\vct{u})\rpa \to {\delta}\lpa \obs{\vct{y}} - \genfunc_y(\vct{u})\rpa$ and so
\begin{align}\label{eq:expectation-input-space-limit}
  \expc{}{f(\rvct{z}) \gvn \rvct{y}=\obs{\vct{y}}} &=
  \lim_{\epsilon\to 0} \lbr \expc{}{f(\rvct{z}) \gvn \obs{\rvct{y}}=\obs{\vct{y}};\,\epsilon} \rbr
  \\
  &\propto
  \int_{\set{U}} 
    f\,\circ\,\genfunc_z(\vct{u}) \,
    \delta\lpa\obs{\vct{y}} - \genfunc_y(\vct{u})\rpa\,
    \rho(\vct{u})
  \,\dr\vct{u}.
  \nonumber
\end{align}
The Dirac delta term restricts the integral across the input space $\set{U}$ to an embedded, $M - N_y$ dimensional, implicitly-defined manifold $\set{C} = \fset{\vct{u} \in \set{U} :  \genfunc_y(\vct{u}) - \obs{\vct{y}} = \vct{0}}$. By applying the the \emph{Co-Area Formula} \citep[\S 3.2.12]{federer2014geometric} the integral with respect to the Lebesgue measure across $\set{U}$ in \eqref{eq:expectation-input-space-limit} can be rewritten as integral across the embedded manifold $\set{C}$ with respect the Hausdorff measure for the manifold
\begin{align}\label{eq:cond-expectation-input-space}
  &\expc{}{f(\rvct{z}) \gvn \rvct{y}=\obs{\vct{y}}}  
  \propto\\
  &
  \int_{\set{C}} 
    f\,\circ\,\genfunc_z(\vct{u}) \,
    {\left|\pd{\genfunc_y}{\vct{u}}\pd{\genfunc_y}{\vct{u}}^{\textsc{t}}\right|^{-\frac{1}{2}}}
    \rho(\vct{u})
  \,\haum{M-N_y}{\vct{u}}.
  \nonumber
\end{align}  
Therefore if we generate a set of input \abbrref{MCMC} samples $\lbrace\vct{u}^{(s)}\rbrace_{s=1}^S$ restricted to $\set{C}$ and with invariant density
\begin{equation}\label{eq:tgt-density-on-manifold}
    \pi(\vct{u}) \propto 
    \left|\pd{\genfunc_y}{\vct{u}}\pd{\genfunc_y}{\vct{u}}^{\textsc{t}}\right|^{-\frac{1}{2}}
    \rho(\vct{u}) 
\end{equation}
with respect to the Hausdorff measure for the manifold, then
\begin{equation}\label{eq:mc-est-cond-expc}
    \expc{}{\func{f}(\rvct{z}) \gvn \rvct{y} = \vct{y}} 
    =
    \lim_{S \to \infty}
    \frac{1}{S} 
    \sum_{s=1}^S \lsb
      \func{f}\,\circ\,\genfunc_z\lpa\vct{u}^{(s)}\rpa
    \rsb.
\end{equation}
Intuitively the determinant term in \eqref{eq:tgt-density-on-manifold} adjusts for the change in the infinitesimal `thickness' (extent in directions orthogonal to the tangent space) of the manifold when mapping through the generator function $\genfunc_y$. The result \eqref{eq:cond-expectation-input-space} is given in a slightly different form in \citep{diaconis2013sampling}. A derivation is provided in Appendix \ref{sec:cond-expc-deriv} of the Supplementary Material.

A general framework for performing asymptotically exact inference in differentiable generative models is therefore to define \abbrref{MCMC} updates which explore the target density \eqref{eq:tgt-density-on-manifold} on the constraint manifold $\set{C}$. We propose here to use a method which simulates the dynamics of a constrained mechanical system.

\section{Constrained Hamiltonian Monte Carlo}\label{sec:chmc}

In \abbrdef{Hamiltonian Monte Carlo}{HMC} \citep{duane1987hybrid,neal2011mcmc} the vector variable of interest $\rvct{u}$ is augmented with a momentum variable $\rvct{p} \in \set \reals^M$. The momentum is taken to be independently Gaussian distributed with zero mean and covariance $\mtx{M}$, often called the mass matrix. The new joint target density is then proportional to $\exp\lsb -H(\vct{u}, \vct{p})\rsb$ where 
\begin{equation}\label{eq:hamiltonian}
    H(\vct{u}, \vct{p}) = 
    -\log\pi(\vct{u}) + \frac{1}{2}\vct{p}\tr\mtx{M}^{-1}\vct{p}
\end{equation}
is termed the \emph{Hamiltonian}.

The canonical Hamiltonian dynamic is described by the system of ordinary differential equations
\begin{equation}\label{eq:hamiltonian_dynamics}
    \td{\vct{u}}{t} = \pd{H}{\vct{p}} = \mtx{M}^{-1}\vct{p},
    ~~
    \td{\vct{p}}{t} = -\pd{H}{\vct{u}} = \pd{\log\pi}{\vct{u}}.
\end{equation}
This dynamic is time-reversible, measure-preserving and exactly conserves the Hamiltonian. Symplectic integrators allow approximate integration of the Hamiltonian flow while maintaining the time-reversibility and measure-preservation properties. Subject to stability bounds on the time-step, such integrators will exactly conserve some `nearby' Hamiltonian, and so the change in the Hamiltonian will tend to remain small even over long simulated trajectories \citep{leimkuhler2004simulating}. 

These properties make simulated Hamiltonian dynamics an ideal proposal mechanism for a Metropolis \abbrref{MCMC} method. The Metropolis accept ratio for a proposal $(\vct{u}_{\text{p}},\,\vct{p}_{\text{p}})$ generated by simulating the dynamic $N_s$ time steps forward from $(\vct{u},\, \vct{p})$ with a symplectic integrator and then negating the momentum, is simply $\exp\lbr H(\vct{u},\,\vct{p}) - H(\vct{u}_{\text{p}},\,\vct{p}_{\text{p}})\rbr$. Typically the change in the Hamiltonian will be small and so the probability of acceptance high. To ensure ergodicity, dynamic moves can be interspersed with updates independently sampling a new momentum from $\nrm{\vct{0},\,\mtx{M}}$.

In our case the system is subject to a constraint of the form $\uconfunc(\vct{u}) = \genfunc_y(\vct{u}) - \obs{\vct{y}} = \vct{0}$. By introducing Lagrangian multipliers $\lambda_i$ for each of the constraints, the Hamiltonian for a constrained system can be written as
\begin{equation}\label{eq:constrained_hamiltonian}
    H(\vct{u}, \vct{p}) = 
    -\log\pi(\vct{u}) + \frac{1}{2}\vct{p}\tr\mtx{M}^{-1}\vct{p}
    + \vct{\lambda}\tr\uconfunc(\vct{u}),
\end{equation}
and a corresponding constrained Hamiltonian dynamic
\begin{align}\label{eq:constrained_hamiltonian_dynamics}
    \td{\vct{u}}{t} = \mtx{M}^{-1}\vct{p},
    ~~
    \td{\vct{p}}{t} = \pd{\log\pi}{\vct{u}} - \pd{\uconfunc}{\vct{u}}\tr\vct{\lambda},
    \\
    \text{subject to }
    \uconfunc(\vct{u}) = \vct{0},
    ~~
    \pd{\uconfunc}{\vct{u}}\mtx{M}^{-1}\vct{p} = \vct{0}.\label{eq:chmc-conditions}
\end{align}

A popular numerical integrator for simulating constrained Hamiltonian dynamics is \abbrref{RATTLE} \citep{andersen1983rattle} (and the algebraically equivalent \abbrref{SHAKE} \citep{ryckaert1977numerical} scheme). This a natural generalisation of the St\"{o}rmer-Verlet (leapfrog) integrator typically used in standard \abbrref{HMC} with additional projection steps in which the Lagrange multipliers $\vct{\lambda}$ are solved for to satisfy the conditions \eqref{eq:chmc-conditions}. RATTLE and SHAKE maintain the properties of being time-reversible, measure-preserving and symplectic \citep{leimkuhler1994symplectic}.

The use of constrained dynamics in \abbrref{HMC} has been proposed several times. In the molecular dynamics literature, both \citep{hartmann2005constrained} and \citep{lelievre2012langevin} suggest using a simulated constrained dynamic within a \abbrref{HMC} framework to estimate free-energy profiles.

Most relevantly here \citep{brubaker2012family} proposes using a constrained \abbrref{HMC} variant to perform inference in distributions defined on implicitly defined embedded non-linear manifolds. This gives sufficient conditions on $H$, $\set{C}$ and $\uconfunc$ for the scheme to satisfy detailed balance and be ergodic: that $H$ is $C^2$ continuous, and $\set{C}$ is a connected smooth and differentiable manifold and $\pd{\uconfunc}{\vct{u}}$ has full row-rank everywhere.

\section{Method}\label{sec:method}

Our constrained \abbrref{HMC} implementation is shown in algorithm \ref{alg:constrained_hmc}. We use a generalisation of the \abbrref{RATTLE} scheme to simulate the dynamic. The inner updates of the state to solve for the geodesic motion on the constraint manifold are split into multiple smaller steps, which can be considered a special case of the scheme described in \citep{leimkuhler2016efficient}. This allows more flexibility in choosing an appropriately small step-size to ensure convergence of the iterative solution of the equations projecting on to the constraint manifold while still allowing a more efficient larger step size for updates to the momentum due to the negative log density gradient. We have assumed $\mtx{M} = \mtx{I}$ here; other mass matrix choices can be equivalently implemented by adding an initial linear transformation stage in the generator.

\begin{algorithm*}
\caption{Constrained Hamiltonian Monte Carlo algorithm}
\label{alg:constrained_hmc}
{\setstretch{1.05}\small
\newcommand{\negenergy}{\log\pi} 
\begin{algorithmic}
\small
    \Require
    $(\vct{u},\, \vct{p})$ : current state pair with $\uconfunc(\vct{u}) = \vct{0},\,\pd{\uconfunc}{\vct{u}}\vct{p} = \vct{0}$; 
    $(\mtx{J},\, \mtx{L})$ : constraint Jacobian and Gram matrix Cholesky factor at current $\vct{u}$;\\ 
    $\epsilon$ : convergence tolerance; 
    $\delta t$ : time step; 
    $N_s$ : number of time steps to simulate; 
    $N_g$ : number of geodesic updates per time step.
    \Ensure\raggedright
    $(\vct{u}',\, \vct{p}')$ : new state pair with $\uconfunc(\vct{u}') = \vct{0}$, $\pd{\uconfunc}{\vct{u}'}\vct{p}' = \vct{0}$; 
    $(\mtx{J}',\, \mtx{L}')$ : constraint Jacobian and Gram matrix Cholesky factor at new $\vct{u}'$.
\end{algorithmic}
\vspace{2mm}
\hrule
\begin{multicols}{2}
\small
\begin{algorithmic}[1]
    \State $\vct{u}_{\text{p}},\,\vct{p}_{\text{p}},\,\mtx{J}_{\text{p}},\,\mtx{L}_{\text{p}} \gets \Call{SimulateDynamic}{\vct{u},\,\vct{p},\,\mtx{J},\,\mtx{L}}$
    \vspace{0.5mm}
    \State $r \gets \textrm{Uniform}(0,\,1)$
    \If{$r < \exp\lbr H(\vct{u},\vct{p}) - H(\vct{u}_{\text{p}},\vct{p}_{\text{p}}) \rbr$}
        \State $\vct{u}',\,\vct{p}',\,\mtx{J}',\,\mtx{L}' 
        \gets \vct{u}_{\text{p}},\, \vct{p}_{\text{p}},\,\mtx{J}_{\text{p}},\,\mtx{L}_{\text{p}}$
    \Else
        \State $\vct{u}',\,\vct{p}',\,\mtx{J}',\,\mtx{L}' 
        \gets \vct{u},\, \vct{p},\,\mtx{J},\,\mtx{L}$
    \EndIf
    \State $\vct{n} \gets \textrm{Normal}(\vct{0},\,\mtx{I})$
    \State $\vct{p}' \gets \Call{ProjectMomentum}{\vct{n},\,\mtx{J}',\,\mtx{L}'}$
    \State 
    \Function{SimulateDynamic}{$\vct{u}$, $\vct{p}$, $\mtx{J}$, $\mtx{L}$}
        \State $\tilde{\vct{p}} \gets \vct{p} + \frac{\delta t}{2} \left.\pd{\negenergy}{\vct{u}}\right|_{\vct{u}}$
        \vspace{0.5mm}
        \State $\vct{p} \gets \Call{ProjectMomentum}{\tilde{\vct{p}},\,\mtx{J},\,\mtx{L}}$
        \State $\vct{u},\,\vct{p},\,\mtx{J},\,\mtx{L} \gets \Call{SimulateGeodesic}{\vct{u},\,\vct{p},\,\mtx{J},\,\mtx{L}}$
        \For{$s \in \fset{2 \dots N_s}$}
            \State $\tilde{\vct{p}} \gets \vct{p} + \delta t \left.\pd{\negenergy}{\vct{u}}\right|_{\vct{u}}$
            \vspace{0.5mm}
            \State $\vct{p} \gets \Call{ProjectMomentum}{\tilde{\vct{p}},\,\mtx{J},\,\mtx{L}}$
            \State $\vct{u},\,\vct{p},\,\mtx{J},\,\mtx{L} \gets \Call{SimulateGeodesic}{\vct{u},\,\vct{p},\,\mtx{J},\,\mtx{L}}$
        \EndFor
        \State $\tilde{\vct{p}}\gets \vct{p} + \frac{\delta t}{2} \left.\pd{\negenergy}{\vct{u}}\right|_{\vct{u}}$
        \vspace{0.5mm}
        \State $\vct{p} \gets \Call{ProjectMomentum}{\tilde{\vct{p}},\,\mtx{J},\,\mtx{L}}$
        \State \Return $\vct{u},\,\vct{p},\,\mtx{J},\,\mtx{L}$
    \EndFunction
    \columnbreak
    \Function{ProjectPosition}{$\vct{u}$, $\mtx{J}$, $\mtx{L}$}
         \State $\vct{c} \gets \uconfunc(\vct{u})$
         \While{$\Vert\vct{c}\Vert_{\infty} > \epsilon$}
             \State $\vct{u} \gets \vct{u} - \mtx{J}\tr \mtx{L}^{-\textrm{T}} \mtx{L}^{-1} \vct{c}$
             \State $\vct{c} \gets \uconfunc(\vct{u})$
         \EndWhile
         \State \Return $\vct{u}$
    \EndFunction
    \State 
    \Function{ProjectMomentum}{$\vct{p}$, $\mtx{J}$, $\mtx{L}$}
        \State \Return $\vct{p} - \mtx{J}\tr\mtx{L}^{-\textrm{T}}\mtx{L}^{-1}\mtx{J}\vct{p}$
    \EndFunction
    \State 
    \Function{SimulateGeodesic}{$\vct{u}$, $\vct{p}$, $\mtx{J}$, $\mtx{L}$}
        \For{$i \in \fset{1 \dots N_g}$}
            \State $\tilde{\vct{u}} \gets \vct{u} + \frac{\delta t}{N_g} \,\vct{p}$
            \State $\vct{u}' \gets \Call{ProjectPosition}{\tilde{\vct{u}}, \mtx{J},\, \mtx{L}}$
            \State $\mtx{J} \gets \left.\pd{\uconfunc}{\vct{u}}\right|_{\vct{u}'}$
            \State $\mtx{L} \gets \chol \lpa \mtx{J}\mtx{J}\tr \rpa$ \label{ln:chmc-cholesky}
            \State $\tilde{\vct{p}} \gets \frac{N_g}{\delta t}\lpa \vct{u}' - \vct{u} \rpa$
            \State $\vct{p} \gets \Call{ProjectMomentum}{\tilde{\vct{p}},\,\mtx{J},\,\mtx{L}}$
            \State $\vct{u} \gets \vct{u}'$
        \EndFor        
        \State \Return $\vct{u},\,\vct{p},\,\mtx{J},\,\mtx{L}$
    \EndFunction
\end{algorithmic}
\end{multicols}
}
\end{algorithm*}

Each inner geodesic time-step involves making an unconstrained update $\vct{u} \rightarrow \tilde{\vct{u}}$ and then projecting $\tilde{\vct{u}}$ back on to $\set{C}$ by solving for $\vct{\lambda}$ which satisfy $\uconfunc(\tilde{\vct{u}} - {\pd{\uconfunc}{\vct{u}}}^\textsc{t}\vct{\lambda}) = \vct{0}$. This is performed in the function \textsc{ProjectPosition} in algorithm \ref{alg:constrained_hmc}. Here we use a quasi-Newton method for solving the system of equations in the projection step. The true Newton update would be
\begin{equation*}\label{eq:newton_iteration}\textstyle
    \vct{u}^{(t)} \leftarrow \vct{u}^{(t)} - 
    \left.\pd{\uconfunc}{\vct{u}}\right|_{\vct{u}^{(t-1)}}^\textsc{t}
    \lbr 
        \left.\pd{\uconfunc}{\vct{u}}\right|_{\vct{u}^{(t)}}
        \left.\pd{\uconfunc}{\vct{u}}\right|_{\vct{u}^{(t-1)}}^\textsc{t}
    \rbr^{-1}
    \uconfunc(\vct{u}^{(t)}).
\end{equation*}
This requires recalculating the Jacobian and solving a dense linear system within the optimisation loop. Instead we use a symmetric quasi-Newton update as proposed in \citep{barth1995algorithms}, the Jacobian Gram matrix $\pd{\uconfunc}{\vct{u}}\pd{\uconfunc}{\vct{u}}^\textsc{t}$ evaluated at the previous state used to condition the moves. This matrix is positive-definite and a Cholesky decomposition can be calculated outside the optimisation loop allowing cheaper quadratic cost solves within the loop. In the rare cases where the quasi-Newton iteration fails we fall back to a \textsc{minpack} \citep{more1980user} implementation of the robust Powell's Hybrid method \citep{powell1970hybrid}.

For larger systems, the Cholesky decomposition of the constraint Jacobian Gram matrix (line \ref{ln:chmc-cholesky}) will become a dominant cost, generally scaling cubically with $N_y$. The elementwise or autoregressive noise structures of many models however allows a significantly reduced quadratically scaling computational cost as explained in the supplementary material in Appendix \ref{sec:generator-structure-cholesky}.

The momentum updates in the \textsc{SimulateDynamic} routine require evaluating the gradient of the logarithm of the target density \eqref{eq:tgt-density-on-manifold}. This can by achieved by using automatic differentiation to calculate the gradient from the expression given in \eqref{eq:tgt-density-on-manifold}, however both the log-density and gradient can be more efficiently calculated by reusing the Cholesky decomposition of the constraint Jacobian Gram matrix computed in line \ref{ln:chmc-cholesky}. Details are given in Appendix \ref{sec:tgt-density-grad}.

In the \textsc{ProjectPosition} routine convergence is signalled when the elementwise maximum absolute value of the constraint function is below some tolerance $\epsilon$. This acts analogously to the $\epsilon$ parameter in \abbrref{ABC} methods, however here we typically set this parameter to some multiple of machine precision and so the approximation introduced is comparable to that otherwise incurred for using non-exact arithmetic.

A final implementation detail is the requirement to find some initial $\vct{u}$ satisfying $\uconfunc(\vct{u}) = \vct{0}$. In some cases structure in the generator function $\genfunc_y$ such as that described in Appendix \ref{sec:generator-structure-cholesky} can be leveraged to come up with an efficient non-iterative scheme for finding a constraint satisfying $\vct{u}$. For general generators, we can choose a subset of the inputs (or linear projections of the inputs) of dimensionality $N_y$ and plug the resulting system of equations into a black-box solver.

\section{Related work}\label{sec:related-work}

Closely related is the \emph{Constrained \abbrref{HMC}} method of \citep{brubaker2012family}. This demonstrates the validity of the constrained \abbrref{HMC} framework theoretically and experimentally, and we do not claim any original contribution in this respect. The focus in \citep{brubaker2012family} is on performing inference in distributions inherently defined on a fixed non-Euclidean manifold such as the unit sphere or space of orthogonal matrices.

Our work builds on \citep{brubaker2012family} by highlighting that conditioning on the output of a generator imposes a constraint on its inputs and so defines a density in input space restricted to some manifold. Unlike the cases considered in \citep{brubaker2012family} our constraints are therefore data-driven and the target density on the manifold implicitly defined by a generator function and base density.

\emph{Geodesic Monte Carlo} \citep{byrne2013geodesic} also considers applying a \abbrref{HMC} scheme to sample from non-linear manifolds embedded in a Euclidean space. Similarly to \citep{brubaker2012family} however the motivation is performing inference with respect to distributions explicitly defined on a manifold such as directional statistics. 

The method presented in \citep{byrne2013geodesic} uses an exact solution for the geodesic flow on the manifold. Our use of constrained Hamiltonian dynamics, and in particular the geodesic integration scheme of \citep{leimkuhler2016efficient}, can be considered an extension for cases when an exact geodesic solution is not available. Instead the geodesic flow is approximately simulated while still maintaining the required measure-preservation and reversibility properties for validity of the overall \abbrref{HMC} scheme.

\emph{Hamiltonian ABC} \cite{meeds2015hamiltonian}, also proposes applying \abbrref{HMC} to perform inference in simulator models as considered here. An {ABC} set-up is used with a Gaussian synthetic-likelihood formed by estimating moments from simulated data.

Rather than using automatic differentiation to exactly calculate gradients of the generator function, \emph{Hamiltonian ABC} uses a stochastic gradient estimator. This is based on previous work considering methods for using a stochastic gradients within \abbrref{HMC} \citep{welling2011bayesian,chen2014stochastic}. It has been suggested however that the use of stochastic gradients can destroy the favourable properties of Hamiltonian dynamics which enable coherent exploration of high dimensional state spaces \citep{betancourt2015fundamental}. 

In \emph{Hamiltonian ABC} it is also observed that representing the generative model as a deterministic function by fixing the random inputs to the generator is a useful method for improving exploration of the state space. This is achieved by including the seed of the pseudo-random number generator in the chain state rather than the set of random inputs.

Also related is \emph{Optimisation Monte Carlo} \citep{meeds2015optimization}. The authors propose using an optimiser to find parameters of a simulator model consistent with observed data (to within some tolerance $\epsilon$) given fixed random inputs sampled independently. The optimisation is not measure-preserving and so the Jacobian of the map is approximated with finite differences to weight the samples. Our method also uses an optimiser to find inputs consistent with the observations, however by using a measure-preserving dynamic we avoid having to re-weight samples which can scale poorly with dimensionality. 

Our method also differs in treating all inputs to a generator equivalently; while the \emph{Optimisation Monte Carlo} authors similarly identify the simulator models as deterministic functions they distinguish between parameters and random inputs, optimising the first and independently sampling the latter. This can lead to random inputs being sampled for which no parameters can be found consistent with the observations (even with a `soft' within $\epsilon$ constraint). Although optimisation failure is also potentially an issue for our method, as we jointly optimise all inputs and are approximating some exact continuous time constraint-satisfying dynamic we found this occurred rarely in practice if an appropriate step size is chosen. Our method can also be applied in cases were the parameter dimension is greater than the dimension of the observations unlike \emph{Optimization Monte Carlo}.

\section{Experiments}\label{sec:experiments}

To illustrate the general applicability of our method we performed inference tasks in three diverse settings: parameter inference in a stochastic Lotka-Volterra predator-prey model simulation, 3D human pose and camera parameter inference given 2D joint position information and finally in-painting of missing regions of digit images using a generative model trained on \abbrref{MNIST}. In all three experiments Theano \citep{theano2016theano} was used to specify the generator function and calculate the required derivatives. All experiments were run on an Intel Core i5-2400 quad-core CPU. Python code for the experiments is available at \url{https://git.io/dgm}.

\subsection{Lotka--Volterra parameter inference}

\begin{figure*}[t]
\centering
\begin{subfigure}[b]{0.217\textwidth}
  \includegraphics[width=\textwidth]{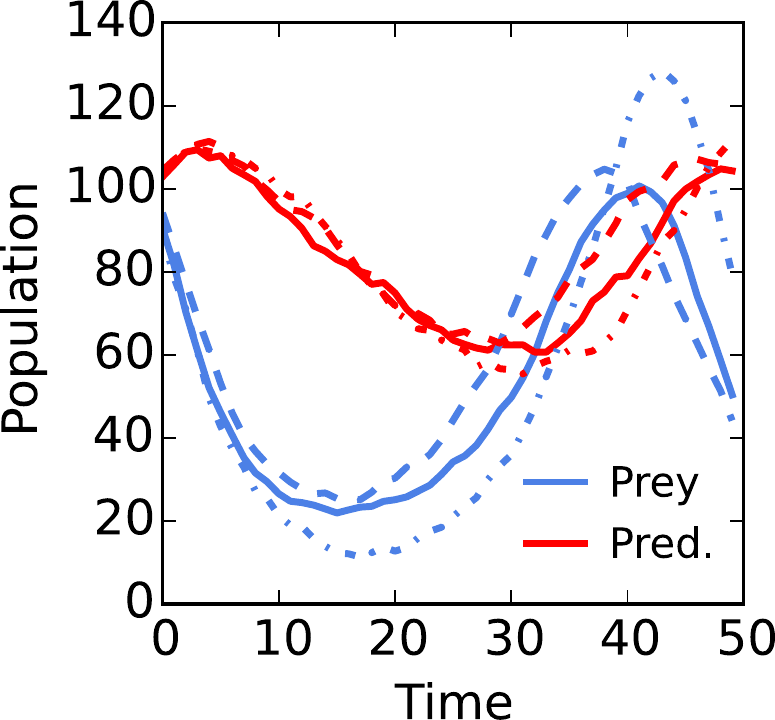}
  \caption{}
  \label{sfig:lotka-volterra-sims}
\end{subfigure}
~~
\begin{subfigure}[b]{0.345\textwidth}
  \includegraphics[width=\textwidth]{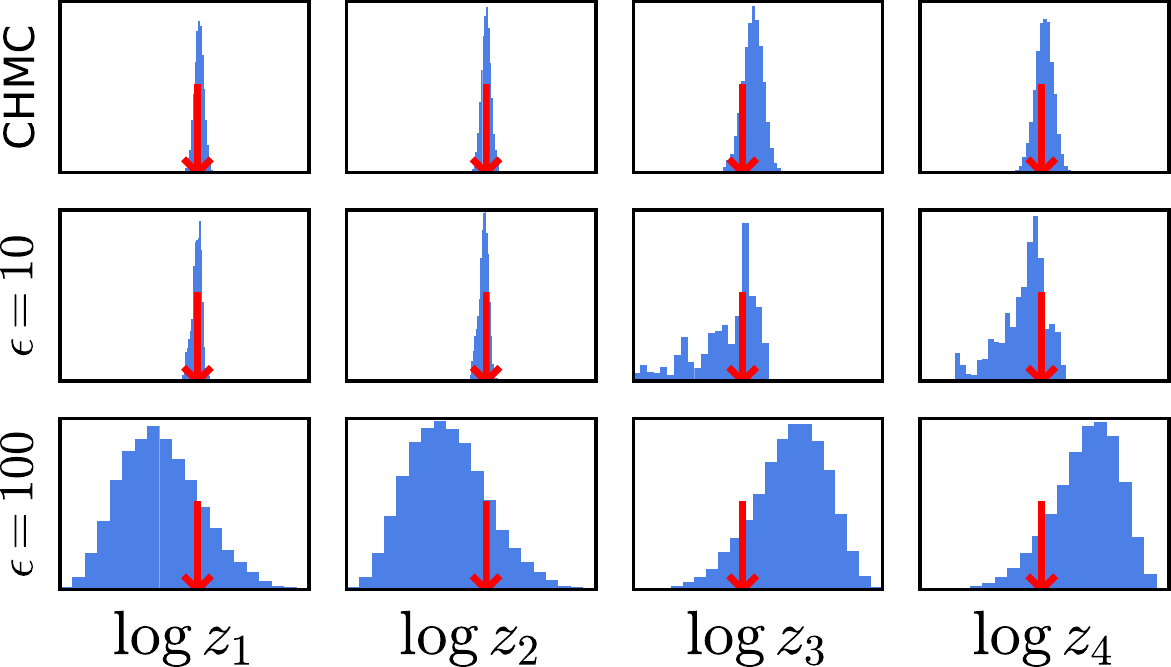}
  \caption{}
  \label{sfig:lotka-volterra-marginals}
\end{subfigure}
~~
\begin{subfigure}[b]{0.347\textwidth}
  \includegraphics[width=\textwidth]{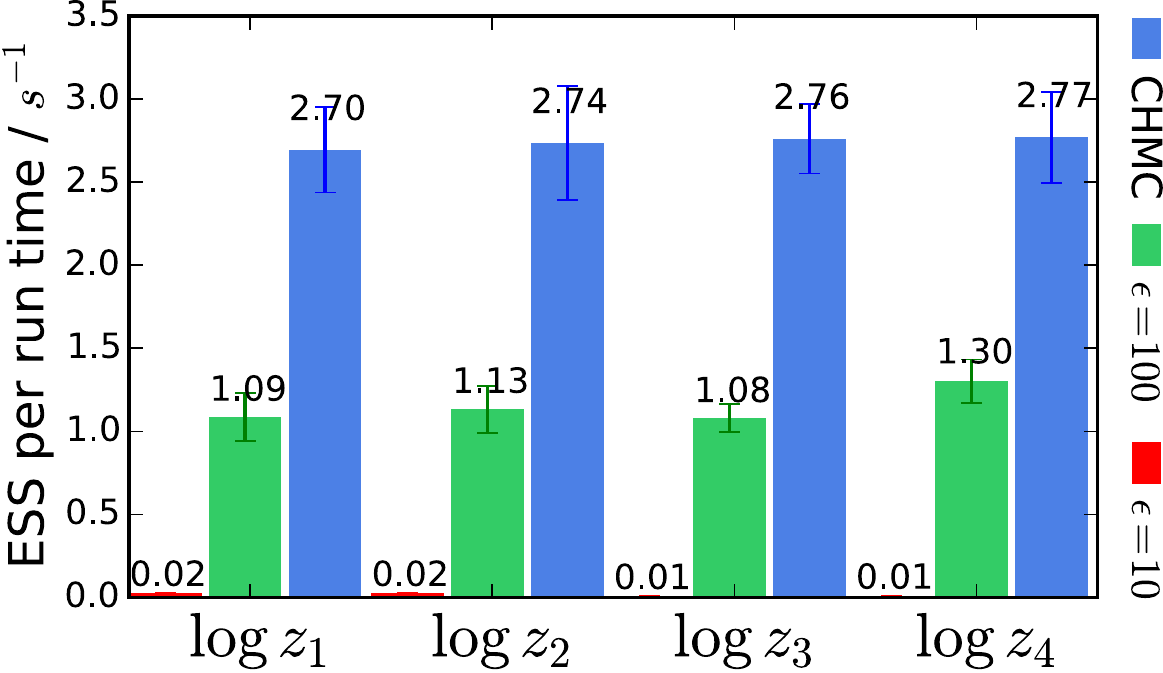}
  \caption{}
  \label{sfig:lotka-volterra-ess}
\end{subfigure}
\caption{\textbf{Lotka--Volterra}~~ (a) Observed predator-prey populations (solid) and \abbrref{ABC} sample trajectories with $\epsilon=10$ (dashed) and $\epsilon=100$ (dot-dashed).
(b) Marginal empirical histograms for the (logarithm of the) four parameters (columns) from constrained \abbrref{HMC} samples (top) and \abbrref{ABC} samples with $\epsilon=10$ (middle) and $\epsilon=100$ (bottom). Horizontal axes shared across columns. Red arrows indicate true parameter values. (c) Mean \abbrref{ESS} normalised by compute time for each of four parameters for \abbrref{ABC} with $\epsilon=10$ (red), $\epsilon=100$ (green) and our method (blue). Error bars show $\pm 3$ standard errors of mean.}
\label{fig:lotka-volterra}
\end{figure*}

As a first demonstration we considered a stochastic continuous state variant of the Lotka--Volterra model, a common example problem for \abbrref{ABC} methods e.g. \citep{meeds2015optimization}. In particular we consider parameter inference given a simulated solution of the following stochastic differential equations
\begin{align}\label{eq:lotka-volterra}
    \dr y_1 &= \lpa z_1 y_1 - z_2 y_1 y_2 \rpa \dr t + \dr n_1,
    \\
    \dr y_2 &= \lpa -z_3 y_2 + z_4 y_1 y_2 \rpa \dr t + \dr n_2,
\end{align}
where $y_1$ represents the prey population, $y_2$ the predator population, $\fset{z_i}_{i=1}^4$ the system parameters and $n_1$ and $n_2$ zero-mean, unit variance white noise processes.

The observed data was generated with an Euler-Maruyama discretisation, time-step $1$, initial condition $y_1^{(0)} = y_2^{(0)} = 100$ and $z_1=0.4$, $z_2 = 0.005$,  $z_3=0.05$, $z_4=0.001$ (chosen to give stable dynamics). The simulation was run for $50$ time-steps with the observed outputs defined as the concatenated vector $\vct{y} = {\scriptstyle \lsb y_1^{(1)} ~ y_2^{(1)} ~ \dots ~ y_1^{(50)} ~ y_2^{(50)} \rsb}$. Log-normal priors $z_i \sim \log \nrm{-2, 1}$ were place on the system parameters.

We compared our method to various \abbrref{ABC} approaches (\S\ref{sec:abc}) using a uniform ball kernel with radius $\epsilon$. \abbrref{ABC} rejection failed catastrophically, with no acceptances in $10^6$ samples even with a large $\epsilon = 1000$. \abbrref{ABC} \abbrref{MCMC} with a Gaussian proposal distribution $q$ also performed very poorly with the dynamic having zero acceptances over multiple runs of $10^5$ updates for $\epsilon=100$ and getting stuck at points in parameter space over many updates for larger $\epsilon=1000$, even with small proposal steps. Based on a method proposed in the pseudo-marginal literature \citep{murray2015pseudo}, we tried using alternating elliptical slice sampling updates of the random inputs used to generate the parameters $\rvct{u}_z$ and remaining random inputs $\rvct{u}_y$ used to generate the simulated observations given parameters. The slice sampling updates locally adapt the size of steps made to ensure a move can always be made. Using this method we were able to obtain reasonable convergence over long runs for both $\epsilon=100$ and $\epsilon=10$.

The results are summarised in Figure \ref{fig:lotka-volterra}. Figure \ref{sfig:lotka-volterra-sims} shows the simulated data used as observations and \abbrref{ABC} sample trajectories for $\epsilon=10$ and $\epsilon=100$ . Though both samples follow the general trends of the observed data there are large discrepancies particularly for $\epsilon=100$. Our method in contrast samples parameters generating trajectories \emph{exactly} matching the observations at all points. Figure \ref{sfig:lotka-volterra-marginals} shows the marginal histograms for the parameters. The inferred posterior on the parameters are significantly more tightly distributed about the true values used to generate the observations for our approach and the $\epsilon=10$ case compared to the results for $\epsilon=100$; even for the $\epsilon=10$ case however it can be seen that there are spurious appearing peaks in the distributions for $z_3$ and $z_4$.

Figure \ref{sfig:lotka-volterra-ess} shows the relative sampling efficiency of our approach against the \abbrref{ABC} methods, as measured by the \abbrdef{effective sample sizes}{ESS} (computed with R-CODA \citep{plummer2006coda}) normalised by run time averaged across 10 sampling runs for each method. Despite the significantly higher per-update cost in our method, the coherent movement about the state space afforded by the Hamiltonian dynamic gave significantly better performance even over the very approximate $\epsilon=100$ case. 

\subsection{Human pose and camera model inference}

\begin{figure*}[t]
\centering
\begin{subfigure}[b]{0.58\textwidth}
  \includegraphics[width=\textwidth]{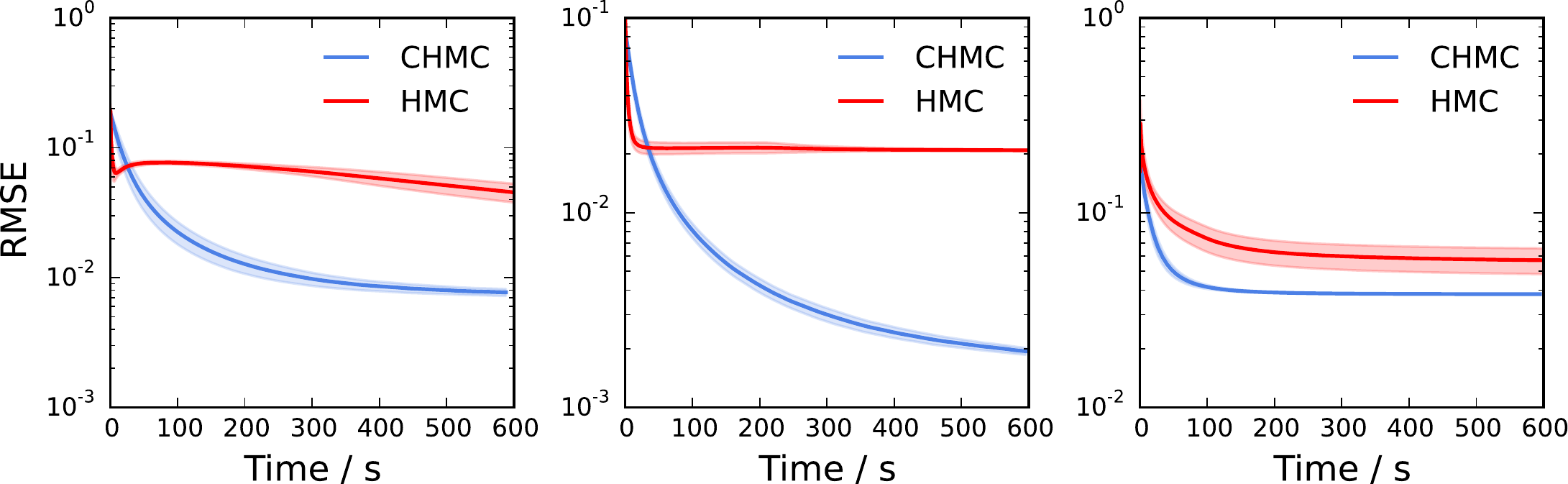}
  \caption{}
  \label{sfig:pose-binocular-rmses}
\end{subfigure}
~~
\begin{subfigure}[b]{0.35\textwidth}
  \includegraphics[width=\textwidth]{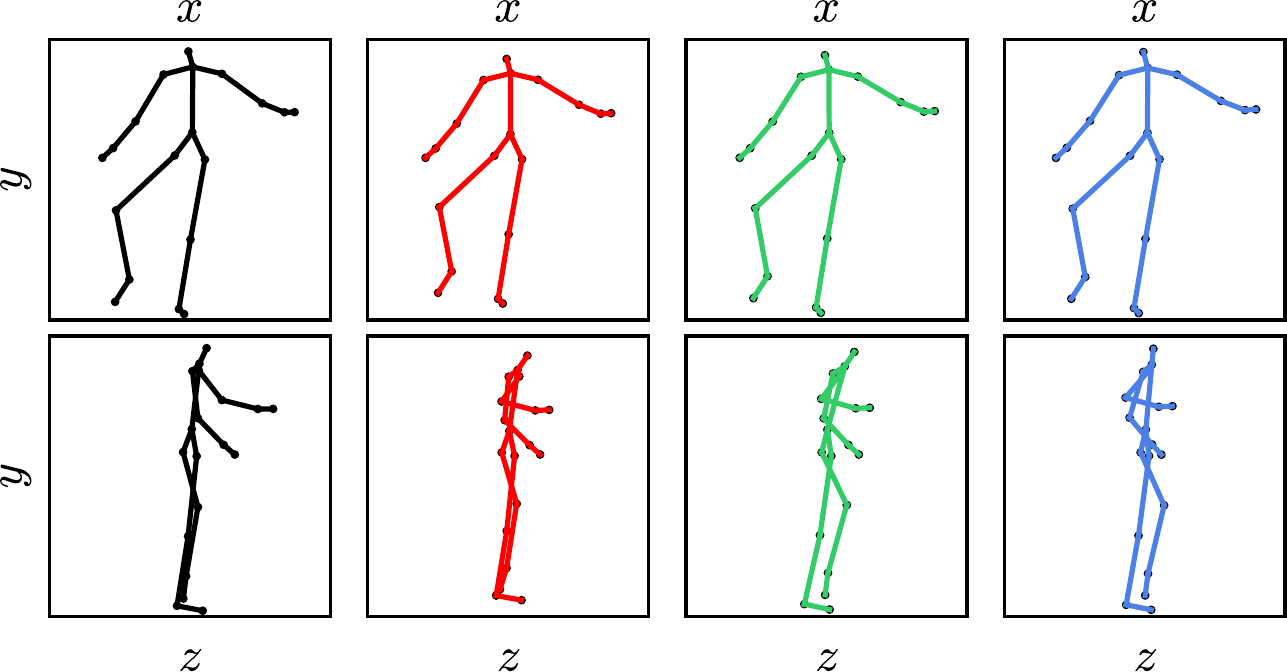}
  \caption{}
  \label{sfig:pose-monocular-samples}
\end{subfigure}
\caption{\textbf{Human pose}~ (a) \abbrref{RMSE}s of 3D pose posterior mean estimates given binocular projections, using samples from our method (blue) versus running \abbrref{HMC} in hierarchical model (red) for three different scenes sampled from the prior. Horizontal axes show computation time to produce number of samples in estimate. Solid curves are average \abbrref{RMSE} over 10 runs with different seeds and shaded regions show $\pm 3$ standard errors of mean. (b) Orthographic projections (top: front view, bottom: side view) of 3D poses consistent with monocular projections. Left most pair (black) shows pose used to generate observations, right three show constrained \abbrref{HMC} samples.}
\label{fig:pose-inference}
\end{figure*}

For our next experiment we considered inferring a 3D human pose and camera model from 2D projection(s) of joint positions. We used a 19 joint skeleton model, with a prior density over poses parametrised by 47 local joint angles $\rvct{z}_a$ learnt from the \emph{PosePrior} motion capture data-set \citep{akhter2015pose} with a \abbrref{VAE} with a 30 dimensional hidden representation $\rvct{h}$. For the bone lengths $\rvct{z}_b$, a log-normal model was fitted to data from the \emph{ANSUR} anthropometric data-set \citep{gordon1988ansur}, due to symmetry 13 independent lengths being specified. A simple pin-hole projective camera model with 3 position parameters $\rvct{z}_c$ and fixed focal-length was used\footnote{The camera orientation was assumed fixed to avoid replicating the degrees of freedom specified by the angular orientation of the root joint of the skeleton: only the relative camera--skeleton orientation is important.}. A log-normal prior was placed on the depth co-ordinate $\rvar{z}_{c,3}$ to enforce positivity with normal priors on the other two co-ordinates $\rvar{z}_{c,1}$ and $\rvar{z}_{c,2}$. A small amount of Gaussian noise was added to the projected positions to give the observed 2D joint positions $\rvct{y}$. This ensured the constraint Jacobian was full row-rank everywhere and gave a known hierarchical joint density on $\fset{\rvct{y},\,\rvct{h},\,\rvct{z}_a,\,\rvct{z}_b,\,\rvct{z}_c}$ allowing comparison with HMC as a baseline.

We first considered binocular pose estimation, with the 3D scene information inferred given 2D projections from two cameras with a known offset between them. In this stereo vision case, the disparity between the projections gives information about the depth direction and so we would expect the posterior distribution on the 3D pose to be tightly distributed around the true values used to generate the observations. We compared our constrained \abbrref{HMC} method to running standard \abbrref{HMC} on the conditional density of $\fset{\rvct{h},\,\rvct{z}_a,\,\rvct{z}_b,\,\rvct{z}_c}$ given $\rvct{y}$.

Figure \ref{sfig:pose-binocular-rmses} shows the \abbrdef{root mean squared error}{RMSE} between the posterior mean estimate of the 3D joint positions and the true positions used to generate the observations as the number of samples included in the estimate increases for three test scenes. For both methods the horizontal axis has been scaled by run time.

The constrained \abbrref{HMC} method (blue curves) tends to give significantly more accurate estimates particularly over longer periods. Visually inspecting the sampled poses and individual run traces (not shown) it seems that the \abbrref{HMC} runs tended to often get stuck in local modes corresponding to a subset of joints being `incorrectly' positioned while still broadly matching the (noisy) projections. The complex dependencies of the joint positions on the angle parameters mean the dynamic struggles to find an update which brings the `incorrect' joints closer to their true positions without moving other joints out of line. The constrained \abbrref{HMC} method seemed to be less susceptible to this issue.

We also considered inferring 3D scene information from a single 2D projection. Monocular projection is inherently information destroying with significant uncertainty to the true pose and camera parameters which generated the observations. Figure \ref{sfig:pose-monocular-samples} shows pairs of orthographic projections of 3D poses: the left most column is the pose used to generate the projection conditioned on and the right three columns are poses sampled using constrained \abbrref{HMC} consistent with the observations. The top row shows front $x$--$y$ views, corresponding to the camera view though with a orthographic rather than perspective projection, the bottom row shows side $z$--$y$ views with the $z$ axis the depth from the camera. The dynamic is able to move between a range of plausible poses consistent with the observations while reflecting the inherent depth ambiguity from the monocular projection.

\subsection{MNIST in-painting}

\begin{figure*}
\centering
\begin{subfigure}[b]{0.48\textwidth}
  \includegraphics[width=\textwidth]{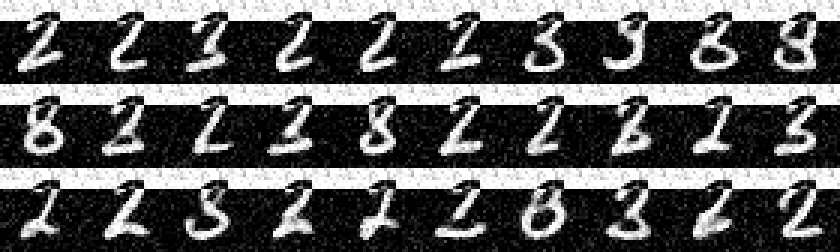}
  \caption{Constrained HMC samples}
  \label{sfig:mnist-samples-chmc}
\end{subfigure}
~~
\begin{subfigure}[b]{0.48\textwidth}
  \includegraphics[width=\textwidth]{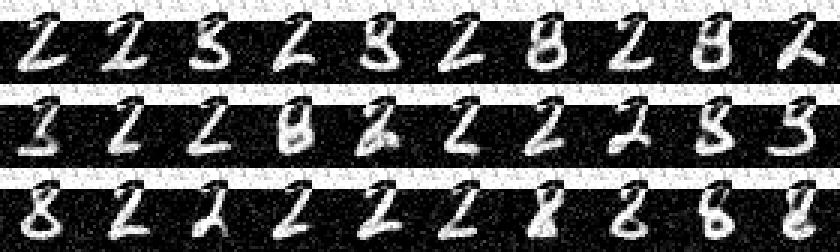}
  \caption{HMC in hierarchical model samples}
  \label{sfig:mnist-samples-hmc}
\end{subfigure}
\caption{\textbf{MNIST}~ In-painting samples. The top black-on-white quarter of each image is the fixed observed region and the remaining white-on-black region the proposed in-painting. In left-right scan order the images in (a) are consecutive samples from a run; in (b) the images are every 40\textsuperscript{th} sample to account for the quicker updates.}
\label{fig:mnist}
\end{figure*}

As a final task we considered inferring in-paintings for a missing region of a digit image $\rvct{z}$ given knowledge of the rest of the pixel values $\rvct{y}$. A Gaussian \abbrref{VAE} trained on \abbrref{MNIST} was used as the generative model, with a 50-dimensional hidden code $\rvct{h}$. We compared our method to running \abbrref{HMC} in the known conditional density on $\rvct{h}$ given $\rvct{y}$ ($\rvct{z}$ can then be directly sampled from its Gaussian conditional given $\rvct{h}$).

Example samples are shown in Figure \ref{fig:mnist}. In this case the constrained and standard \abbrref{HMC} approaches appear to be performing similarly, with both able to find a range of plausible in-paintings given the observed pixels. Without cost adjustment the standard \abbrref{HMC} samples show greater correlation between subsequent updates, however for a fairer comparison the samples shown were thinned to account for the approximately $40\times$ larger run-time per constrained \abbrref{HMC} sample. Although the constrained dynamic does not improve efficiency here neither does it seem to hurt it.

\vspace{-1mm}
\section{Discussion}

We have presented a generally applicable framework for performing inference in differentiable generative models. Though simulating the constrained Hamiltonian dynamic is computationally costly, the resulting coherent exploration of the state space can lead to significantly improved sampling efficiency over alternative methods.

Further our approach allows asymptotically exact inference in differentiable generative models where \abbrref{ABC} methods might otherwise be used. We suggest an approach for dealing with two of the key issues in \abbrref{ABC} methods --- enabling inference in continuous spaces as $\epsilon$ collapses to zero and allowing efficient inference when conditioning on high-dimensional observations without the need for dimensionality reduction with summary statistics (and the resulting task of choosing appropriate summary statistics). As well as being of practical importance itself, this approach should be useful in providing `ground truth' inferences in more complex models to assess the affect of the approximations used in \abbrref{ABC} methods on the quality of the inferences.

In molecular simulations, constrained dynamics are often used to improve efficiency. Intra-molecular motion is removed by fixing bond lengths. This allows a larger time-step to be used due to the removal of high-frequency bond oscillations \citep{leimkuhler2016efficient}. An analogous effect is present when performing inference in an \abbrref{ABC} setting with a $\epsilon$ kernel `soft-constraint' to enforce consistency between the inputs and observed outputs. As $\epsilon \to 0$ the scales over which the inputs density changes value in directions orthogonal to the constraint manifold and along directions tangential to the manifold increasingly differ. To stay within the soft constraint a very small step-size needs to be used. Using a constrained dynamic decouples the motion on the constraint manifold from the steps to project on to it, allowing more efficient larger steps to be used for moving on the manifold.

A limitation of our method is the requirement of differentiability of the generator. This prevents applying our approach to generative models which use discontinuous operations or discrete random inputs. In some cases conditioned on fixed values of discrete random inputs the generator may still be differentiable and so the proposed method can be used to update the continuous random inputs given values of the discrete inputs. This would need to be alternated with updates to the discrete inputs, which would require devising methods for updating the discrete inputs to the generator while constraining its output to exactly match observations.

A common technique in \abbrref{ABC} methods is to define a kernel or distance measure in terms of summary statistics of the observed data \citep{marin2012approximate}. This is necessary in standard \abbrref{ABC} approaches to cope with the `curse of dimensionality' with the probability of accepting samples / moves for a fixed tolerance $\epsilon$ exponentially decreasing as the dimensionality of the observations conditioned on increases. Although as already noted the proposed method is better able to cope with high observation dimensions, if appropriate informative statistics are available (e.g. based on expert knowledge) and these are differentiable functions of the generator outputs, they can be easily integrated in to the proposed method by absorbing the statistic computation in to the definition of $\genfunc_y$.

\subsubsection*{Acknowledgements}

Thanks to Iain Murray for several useful discussions and to Ben Leimkuhler for pointing out the potential of using a geodesic integration scheme. This work was supported in part by grants EP/F500385/1 and  BB/F529254/1 for the University of Edinburgh School of Informatics Doctoral Training Centre in Neuroinformatics and Computational Neuroscience (www.anc.ac.uk/dtc) from the UK Engineering and Physical Sciences Research Council (EPSRC), UK Biotechnology and Biological Sciences Research Council (BBSRC), and the UK Medical Research Council (MRC).  

\pagebreak
\onecolumn
\appendix

\section{Computing conditional expectations in the input space}\label{sec:cond-expc-deriv}

We here give a derivation for the target density on the constraint manifold in the input space of a differentiable generative model for computing expectations conditional on observations of the output. This is largely a restatement of results in \citep[\S 2.3]{diaconis2013sampling} and is provided mainly to make those results more easily relatable to the notation of this paper.

For clarity in this section the measure being integrated with respect to will be explicitly denoted. For a variable of integration $\vct{x}$, $\lebm{D}{\vct{x}}$ will denote the $D$ dimensional Lebesgue measure and $\haum{D}{\vct{x}}$ the $D$ dimensional Hausdorff measure over some space $\set{X}$ which will be specified.

A key result we will use is Federer's \emph{Co-Area Formula} \citep[\S 3.2.12]{federer2014geometric}:

\begin{theorem}[Co-Area Formula]
Let $\vctfunc{m} : \set{X} \subseteq \reals^L \to \set{V} \subseteq \reals^K$ be Lipschitz with $L > K$ and $\func{h} : \set{X} \to \reals$ be Lebesgue measurable. Then
\begin{equation}\label{eq:co-area-formula}
    \int_{\set{X}} 
      \func{h}(\vct{x}) \, \left|\pd{\vctfunc{m}}{\vct{x}}\pd{\vctfunc{m}}{\vct{x}}\tr\right|^{\frac{1}{2}}
    \,\lebm{L}{\vct{x}}
    =
    \int_{\set{V}}
      \int_{\vctfunc{m}^{-1}(\vct{v})} \func{h}(\vct{x}) \,\haum{L-K}{\vct{x}}
    \,\lebm{K}{\vct{v}}   
\end{equation}
with $\pd{\vctfunc{m}}{\vct{x}} = \lsb \pd{\func{m}_i}{x_j} \rsb_{i,j}$ the Jacobian of the map, $\vctfunc{m}^{-1}(\vct{v})$ the $L-K$ dimensional sub-manifold embedded in $\set{X}$ with Hausdorff measure $\haum{L-K}{\vct{x}}$, which is the solution set $\fset{\vct{x} \in \set X : \vctfunc{m}(\vct{x}) = \vct{v}}$. 
\end{theorem}
\begin{corollary}
If the Jacobian of $\vctfunc{m}$ has full row-rank everywhere such that $\left|\pd{\vctfunc{m}}{\vct{x}} \pd{\vctfunc{m}}{\vct{x}}\tr\right| > 0 ~\forall \vct{x} \in \set{X}$ then for a Lebesgue measurable $\func{h}^\ast : \set{X} \to \reals$
\begin{equation}\label{eq:co-area-formula-corollary}
    \int_{\set{X}} 
      \func{h}^\ast(\vct{x}) 
    \,\lebm{L}{\vct{x}}
    =
    \int_{\set{V}}
      \int_{\vctfunc{m}^{-1}(\vct{v})} 
        \func{h}^\ast(\vct{x}) \,
        \left|\pd{\vctfunc{m}}{\vct{x}}\pd{\vctfunc{m}}{\vct{x}}\tr\right|^{-\frac{1}{2}}
      \,\haum{L-K}{\vct{x}}
    \,\lebm{K}{\vct{v}}   
\end{equation}
which can be easily shown by setting $\func{h}(\vct{x}) = \func{h}^\ast(\vct{x}) \left|\pd{\vctfunc{m}}{\vct{x}}\pd{\vctfunc{m}}{\vct{x}}\tr\right|^{-\frac{1}{2}}$ in \eqref{eq:co-area-formula}.
\end{corollary}

We will also use what is sometimes termed the \abbrdef{Law of the Unconscious Statistician}{LOTUS} to express expectations of functions of random (vector) variables when an explicit density on the random output of the function is not known.
\begin{theorem}[Law of the Unconscious Statistician]
Let $\rvct{y}$ be a random vector on support $\set{X} \subseteq \reals^N$ with density $\pden{\rvct{x}}{\vct{x}}$ with respect to the Lebesgue measure $\lebm{N}{\vct{x}}$ and $\func{f}: \set{Y} \to \reals$ be Lebesgue measurable. If we define a new random variable $\rvar{v} = \func{f}(\rvct{x})$ then
\begin{equation}\label{eq:lotus}
    \expc{}{\rvar{v}} = \expc{}{\func{f}(\rvct{x})} =
    \int_{\set{X}} \func{f}(\vct{x}) \, \pden{\rvct{x}}{\vct{x}} \,\lebm{N}{\vct{x}}.
\end{equation}
\end{theorem}
\begin{corollary}
If $\rvct{x}$ is defined as $\rvct{x} = \vctfunc{g}(\rvct{u})$ for some $\vctfunc{g} : \set{U} \subseteq \reals^M \to \set{X}$ then
\begin{equation}\label{eq:lotus-corollary}
    \expc{}{\func{f}(\rvct{x})} = \expc{}{(\func{f}\,\circ\,\vctfunc{g})(\rvct{u})} =
    \int_{\set{U}} \func{f}\,\circ\,\vctfunc{g}(\vct{u}) \, \pden{\rvct{u}}{\vct{u}} \,\lebm{M}{\vct{u}}.
\end{equation}
\end{corollary}

This leads us to the main result
\begin{theorem}
Let $\rvct{u}$ be a random vector with density $\pden{\rvct{u}}{\vct{u}} = \rho(\vct{u})$  with respect to the Lebesgue measure $\lebm{M}{\vct{u}}$ on support $\set{U} = \reals^M$. Further let $\genfunc : \set{U} \to \set{X}$ be a smooth map, with $\set{X} = \reals^N;\, N \leq M$ defining a random vector $\rvct{x} = \genfunc(\rvct{u})$. Assume $\pd{\genfunc}{\vct{u}}$ exists and has full row-rank almost everywhere.

Partition the output space $\set{X} = \set{Y} \times \set{Z}$ with $\set{Y} = \reals^{N_y}$ and $\set{Z} = \reals^{N_z}$ and  $\rvct{y} = \genfunc_y(\rvct{u})$, $\rvct{z} = \genfunc_z(\rvct{u})$. Then the conditional expectation of some function $\func{f}$ of $\rvct{z}$ given $\rvct{y} = \obs{\vct{y}}$ has been observed is
\begin{align}\label{eq:cond-expc-u-haum}
    \expc{}{\func{f}(\rvct{z}) \gvn \rvct{y} = \obs{\vct{y}}} 
    &=
    \frac{1}{\pden{\rvct{y}}{\obs{\vct{y}}}} 
    \int_{\set{C}} 
      \func{f}\,\circ\,\genfunc_z(\vct{u}) \, 
      \rho(\vct{u})\,
      \left| \pd{\genfunc_y}{\vct{u}}\pd{\genfunc_y}{\vct{u}}\tr \right|^{-\frac{1}{2}}
    \,\haum{M-N_y}{\vct{u}}.
\end{align}
with $\pden{\rvct{y}}{\obs{\vct{y}}}$ the marginal density on $\rvct{y}$ with respect to the Lebesgue measure $\lebm{N_y}{\vct{y}}$ which must be non-zero for the conditional expectation to be well-defined; $\set{C}$ is the $M - N_y$ dimensional sub-manifold defined by the solution set $\fset{\vct{u} \in \set{U} : \genfunc_y(\vct{u}) = \obs{\vct{y}}}$.
\end{theorem}

\begin{proof}
By the \emph{Law of Total Expectation} we have that
\begin{equation}
    \expc{}{f(\rvct{z})} =
    \int_{\set{Y}}
      \expc{}{f(\rvct{z}) \gvn \rvct{y} = \vct{y}}\,
      \pden{\rvct{y}}{\vct{y}}\,
    \lebm{N_y}{\vct{y}}.
\end{equation}
Using \abbrref{LOTUS} \eqref{eq:lotus} we get
\begin{equation}
    \int_{\set{Y}}
      \expc{}{f(\rvct{z}) \gvn \rvct{y} = \vct{y}}\,
      \pden{\rvct{y}}{\vct{y}}\,
    \lebm{N_y}{\vct{y}}
    =
    \int_{\set{U}}
       \func{f}\,\circ\,\genfunc_z(\vct{u}) \, \rho(\vct{u}) \,
    \lebm{M}{\vct{u}}.
\end{equation}
Applying the co-area formula corollary \eqref{eq:co-area-formula-corollary} to the right-hand side gives
\begin{align}
    &\int_{\set{Y}}
      \expc{}{f(\rvct{z}) \gvn \rvct{y} = \vct{y}}\,
      \pden{\rvct{y}}{\vct{y}}\,
    \lebm{N_y}{\vct{y}} =
    \\
    &\int_{\set{Y}}
      \int_{\genfunc_y^{-1}(\vct{y})}
        \func{f}\,\circ\,\genfunc_z(\vct{u}) \, 
        \rho(\vct{u})\,
        \left| \pd{\genfunc_y}{\vct{u}}\pd{\genfunc_y}{\vct{u}}\tr \right|^{-\frac{1}{2}}\,
      \haum{M-N_y}{\vct{u}}
    \lebm{N_y}{\vct{y}}.
\end{align}
Define $\set{Y}^\star = \fset{\vct{y} \in \set{Y} : \pden{\rvct{y}}{\vct{y}} > 0}$. Then we have
\begin{align}
    &\int_{\set{Y}^\star}
      \lbr \expc{}{f(\rvct{z}) \gvn \rvct{y} = \vct{y}} \rbr\,
      \pden{\rvct{y}}{\vct{y}}\,
    \lebm{N_y}{\vct{y}} =
    \\
    &\int_{\set{Y}^\star}
      \lbr 
        \frac{1}{\pden{\rvct{y}}{\vct{y}}}
        \int_{\genfunc_y^{-1}(\vct{y})}
          \func{f}\,\circ\,\genfunc_z(\vct{u}) \, 
          \rho(\vct{u})\,
          \left| \pd{\genfunc_y}{\vct{u}}\pd{\genfunc_y}{\vct{u}}\tr \right|^{-\frac{1}{2}}\,
        \haum{M-N_y}{\vct{u}}
      \rbr 
      \pden{\rvct{y}}{\vct{y}}
    \lebm{N_y}{\vct{y}}.
\end{align}
As this holds for arbitrary Lebesgue measurable $\func{f}$, this implies that the terms inside the braces are equal for all $\vct{y} \in \set{Y}^\star$. 
As $\obs{\vct{y}} \in \set{Y}^\star$ by assumption and $\set{C} = \genfunc_y^{-1}(\obs{\vct{y}})$ we have
\begin{align}
    \expc{}{\func{f}(\rvct{z}) \gvn \rvct{y} = \obs{\vct{y}}} 
    &=
    \frac{1}{\pden{\rvct{y}}{\obs{\vct{y}}}} 
    \int_{\set{C}} 
      \func{f}\,\circ\,\genfunc_z(\vct{u}) \, 
      \rho(\vct{u})\,
      \left| \pd{\genfunc_y}{\vct{u}}\pd{\genfunc_y}{\vct{u}}\tr \right|^{-\frac{1}{2}}
    \,\haum{M-N_y}{\vct{u}}.
\end{align}
\end{proof}

\begin{corollary}
Define a target density with respect to the Hausdorff measure $\haum{M-N_y}{\vct{u}}$ on $\set{C}$
\begin{equation}\label{eq:cond-expc-tgt-density-haum}
    \pi(\vct{u}) \propto
    \rho(\vct{u})\,
    \left| \pd{\genfunc_y}{\vct{u}}\pd{\genfunc_y}{\vct{u}}\tr \right|^{-\frac{1}{2}}.
\end{equation}
If we generate a set of \abbrref{MCMC} samples $\fset{\vct{u}^{(s)}}_{s=1}^S$ which leave $\pi(\vct{u})$ invariant with respect to $\haum{M-N_y}{\vct{u}}$ on $\set{C}$, by the \emph{Law of Large Numbers} we can then form a Monte Carlo estimate for the conditional expectation 
\begin{equation}\label{eq:mc-est-cond-expc}
    \expc{}{\func{f}(\rvct{z}) \gvn \rvct{y} = \obs{\vct{y}}} 
    =
    \lim_{S \to \infty}
    \frac{1}{S} 
    \sum_{s=1}^S \lpa
      \func{f}\,\circ\,\genfunc_z(\vct{u}^{(s)})
    \rpa.
\end{equation}
\end{corollary}

\section{Evaluating the target density and its gradient}\label{sec:tgt-density-grad}

For the constrained Hamiltonian dynamics we need to be able to evaluate the logarithm of the target density \eqref{eq:cond-expc-tgt-density-haum} up to an additive constant and its gradient with respect to $\vct{u}$. We have that
\begin{equation} \label{eq:log-cond-tgt-density}
    \log\pi(\vct{u}) = 
    \log\rho(\vct{u}) - 
    \frac{1}{2} \log\left| \pd{\genfunc_y}{\vct{u}}\pd{\genfunc_y}{\vct{u}}\tr \right| - 
    \log{Z}
\end{equation}
where $Z$ is the normalising constant for the density which is independent of $\vct{u}$.

In general evaluating the Gram matrix determinant $\log\left| \pd{\genfunc_y}{\vct{u}}\pd{\genfunc_y}{\vct{u}}\tr \right|$ has computational cost which scales as $\mathcal{O}(M N_y^2)$. However as part of the constrained dynamics updates the lower-triangular Cholesky decomposition $\mtx{L}$ of the Gram matrix $\pd{\genfunc_y}{\vct{u}}\pd{\genfunc_y}{\vct{u}}\tr$ is calculated. Using basic properties of the matrix determinant we have
\begin{align}
    \log\pi(\vct{u}) &= 
    \log\rho(\vct{u}) - 
    \frac{1}{2} \log\left| \mtx{L}\mtx{L}\tr \right| - 
    \log{Z} \\
    &=
    \log\rho(\vct{u}) - 
    \frac{1}{2}\log\left|\mtx{L}\right|\left|\mtx{L}\tr\right| -
    \log{Z} \\
    &=
    \log\rho(\vct{u}) - 
    \log\left|\mtx{L}\right| - 
    \log{Z} \\
    &=
    \log\rho(\vct{u}) - 
    \sum_{i=1}^{N_y}\log(L_{ii}) - 
    \log{Z}
    \label{eq:log-cond-tgt-density-chol}
\end{align}
The base density $\rho(\vct{u})$ will typically be of a simple form e.g. standard Gaussian, therefore we can evaluate the logarithm of the target density up to an additive constant at a marginal computational cost that scales linearly with dimensionality.

For the gradient we can use reverse-mode automatic differentiation to calculate the gradient of \eqref{eq:log-cond-tgt-density-chol} with respect to $\vct{u}$. This requires propagating partial derivatives through the Cholesky decomposition \citep{murray2016differentiation}; efficient implementations for this are present in many automatic differentiation frameworks including Theano.

Alternatively the gradient of \eqref{eq:log-cond-tgt-density} can be explicitly derived. The gradient of $\log\rho(\vct{u})$ will generally be trivial and $\pd{\log Z}{\vct{u}} = \vct{0}$. The gradient of the second term can be calculated using
\begin{align}
    \pd{}{u_i} \log\left| \pd{\genfunc_y}{\vct{u}}\pd{\genfunc_y}{\vct{u}}\tr \right| &= 
    \trace\lbr 
      \lsb \pd{\genfunc_y}{\vct{u}}\pd{\genfunc_y}{\vct{u}}\tr \rsb^{-1} 
      \lsb 
        \pdd{\genfunc_y}{u_i}{\vct{u}} \pd{\genfunc_y}{\vct{u}} \tr +
        \pd{\genfunc_y}{\vct{u}} \pdd{\genfunc_y}{u_i}{\vct{u}} \tr
      \rsb
    \rbr\\
    &=
    2 \trace\lbr 
      \pdd{\genfunc_y}{u_i}{\vct{u}}
      \lsb \mtx{L}^{-\rm T}\mtx{L}^{-1} \pd{\genfunc_y}{\vct{u}} \rsb\tr
    \rbr.
\end{align}
The matrix inside the square brackets is independent of $i$ and can be computed once by solving the system of equations by forward and backward substitution. The matrix of second partial derivatives $\pdd{\genfunc_y}{u_i}{\vct{u}}$ can either be manually derived for the specific generator function or calculated using automatic differentiation. The trace of the matrix product is then just the sum over all indices of the element-wise product of the pair.

\section{Exploiting structure in the generator}\label{sec:generator-structure-cholesky}

Often the generator inputs $\rvct{u}$ can be split in to two distinct groups --- global inputs $\rvct{v}$ which effect all of the observed outputs (e.g. inputs which map to model parameters) and local `noise' inputs $\rvct{n}$, each element of which affect only a subset of the outputs. 

In particular systems with a generator function $\genfunc_y$ which can be expressed in one of the two forms
\begin{equation}\label{eq:elem_autoreg_model_struct}
    y_i = \func{g}_i(\vct{v}, n_i)
    ~\text{(element-wise)}
    ~~\text{or}~~
    y_{i} = \tilde{\func{g}}_i(\vct{v}, y_{i-1}, n_i) = \func{g}_i(\vct{v}, \vct{n}_{\leq i})
    ~\text{(autoregressive)}
\end{equation}
have a Jacobian $\pd{\uconfunc}{\vct{u}} = \lsb {\scriptstyle \pd{\uconfunc}{\vct{v}} ~ \pd{\uconfunc}{\vct{n}}} \rsb$ in which $\pd{\uconfunc}{\vct{n}}$ is diagonal (element-wise) or triangular (autoregressive). 

The decomposition of $\pd{\uconfunc}{\vct{u}}\pd{\uconfunc}{\vct{u}}\inlinetr = \pd{\uconfunc}{\vct{n}}\pd{\uconfunc}{\vct{n}}\phantom{}^{\rm T} + \pd{\uconfunc}{\vct{v}} \pd{\uconfunc}{\vct{v}}\phantom{}^{\rm T}$ can then be computed by low-rank Cholesky updates of the triangular / diagonal matrix $\pd{\uconfunc}{\vct{n}}$ with each of the columns of $\pd{\uconfunc}{\vct{v}}$. As $\dim(\vct{v}) = L$ is often significantly less than, and independent of, the number of outputs conditioned on $N_y$, the resulting $\mathcal{O}(LN_y^2)$ cost of the Cholesky updates is a significant improvement over the original $\mathcal{O}(N_y^3)$.

Many learnt differentiable generative models have an element-wise noise structure including the Gaussian \abbrref{VAE}. The autoregressive noise structure commonly occurs in stochastic dynamical simulations where the outputs are a time sequence of states, with noise being added each time-step, for example the Lotka-Volterra model considered in the experiments in Section \ref{sec:experiments}.

\section{Lotka-Volterra parameter empirical histogram}\label{sec:larger-lv-param-hist}

Larger version of figure \ref{sfig:lotka-volterra-marginals} showing empirical histograms for posterior samples of Lotka--Volterra model parameters.

\begin{figure}[H]
\centering
  \includegraphics[width=\textwidth]{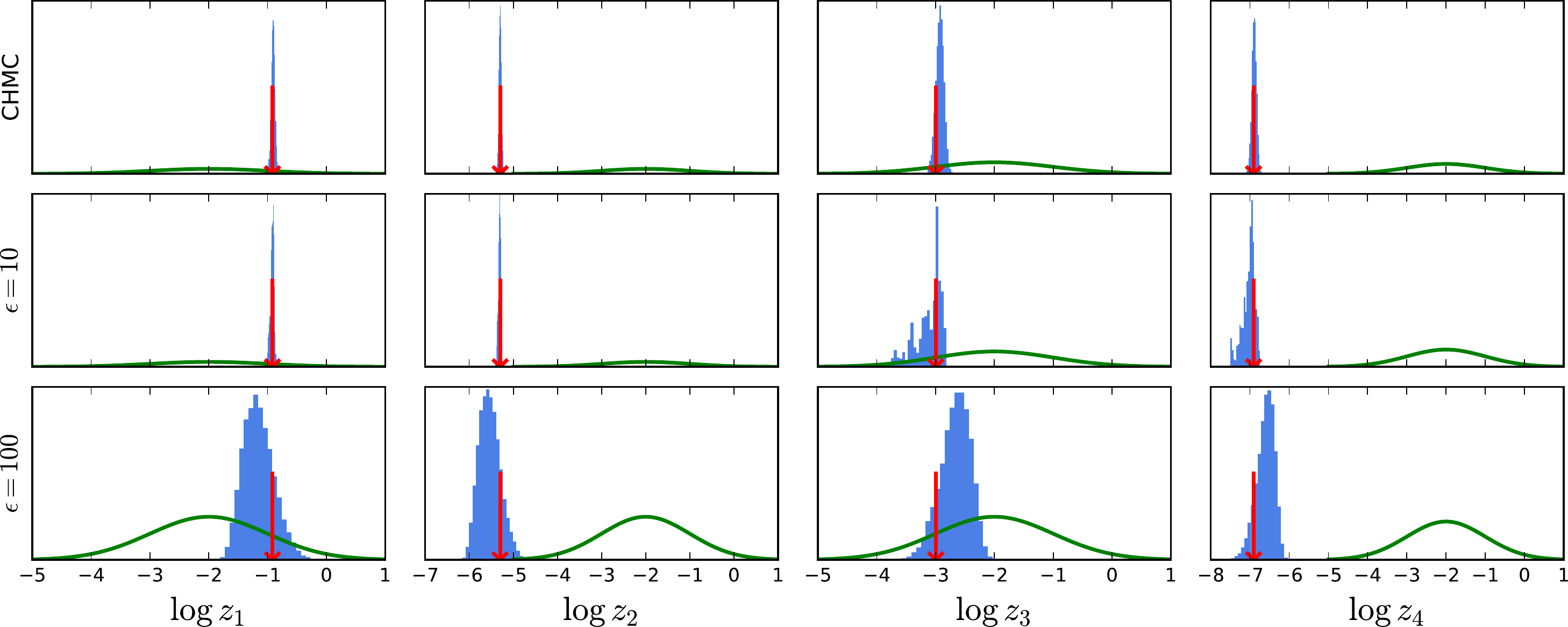}
  \caption{Marginal empirical histograms for the (logarithm of the) four parameters (columns) from constrained \abbrref{HMC} samples (top) and \abbrref{ABC} samples with $\epsilon=10$ (middle) and $\epsilon=100$ (bottom). Horizontal axes shared across columns. Red arrows indicate true parameter values. Green curve - log-normal prior density.}
  \label{fig:lotka-volterra-marginals-with-prior}
\end{figure}

\newpage
\twocolumn
\bibliography{refs}

\begin{thebibliography}{10}

\bibitem{akhter2015pose}
I.~Akhter and M.~J. Black.
\newblock Pose-conditioned joint angle limits for {3D} human pose
  reconstruction.
\newblock In {\em IEEE Conference on Computer Vision and Pattern Recognition},
  2015.

\bibitem{andersen1983rattle}
H.~C. Andersen.
\newblock {RATTLE}: A velocity version of the {SHAKE} algorithm for molecular
  dynamics calculations.
\newblock {\em Journal of Computational Physics}, 1983.

\bibitem{barth1995algorithms}
E.~Barth, K.~Kuczera, B.~Leimkuhler, and R.~D. Skeel.
\newblock Algorithms for constrained molecular dynamics.
\newblock {\em Journal of computational chemistry}, 1995.

\bibitem{beaumont2002approximate}
M.~A. Beaumont, W.~Zhang, and D.~J. Balding.
\newblock Approximate {B}ayesian computation in population genetics.
\newblock {\em Genetics}, 2002.

\bibitem{betancourt2015fundamental}
M.~Betancourt.
\newblock The fundamental incompatibility of scalable {H}amiltonian {M}onte
  {C}arlo and naive data subsampling.
\newblock In {\em Proceedings of the 32nd International Conference on Machine
  Learning}, 2015.

\bibitem{brubaker2012family}
M.~A. Brubaker, M.~Salzmann, and R.~Urtasun.
\newblock A family of {MCMC} methods on implicitly defined manifolds.
\newblock In {\em International Conference on Artificial Intelligence and
  Statistics}, 2012.

\bibitem{byrne2013geodesic}
S.~Byrne and M.~Girolami.
\newblock Geodesic {M}onte {C}arlo on embedded manifolds.
\newblock {\em Scandinavian Journal of Statistics}, 2013.

\bibitem{chen2014stochastic}
T.~Chen, E.~Fox, and C.~Guestrin.
\newblock Stochastic gradient {H}amiltonian {M}onte {C}arlo.
\newblock In {\em Proceedings of the 31st International Conference on Machine
  Learning}, 2014.

\bibitem{diaconis2013sampling}
P.~Diaconis, S.~Holmes, and M.~Shahshahani.
\newblock Sampling from a manifold.
\newblock In {\em Advances in Modern Statistical Theory and Applications},
  pages 102--125. Institute of Mathematical Statistics, 2013.

\bibitem{duane1987hybrid}
S.~Duane, A.~D. Kennedy, B.~J. Pendleton, and D.~Roweth.
\newblock Hybrid {M}onte {C}arlo.
\newblock {\em Physics letters B}, 1987.

\bibitem{federer2014geometric}
H.~Federer.
\newblock {\em Geometric measure theory}.
\newblock Springer, 2014.

\bibitem{goodfellow2014generative}
I.~Goodfellow, J.~Pouget-Abadie, M.~Mirza, B.~Xu, D.~Warde-Farley, S.~Ozair,
  A.~Courville, and Y.~Bengio.
\newblock Generative adversarial nets.
\newblock In {\em Advances in Neural Information Processing Systems}, 2014.

\bibitem{gordon1988ansur}
C.~C. Gordon, T.~Churchill, C.~E. Clauser, B.~Bradtmiller, J.~T. McConville,
  I.~Tebbets, and R.~A. Walker.
\newblock Anthropometric survey of {US} army personell: Final report.
\newblock Technical report, United States Army, 1988.

\bibitem{hartmann2005constrained}
C.~Hartmann and C.~Schutte.
\newblock A constrained hybrid {M}onte-{C}arlo algorithm and the problem of
  calculating the free energy in several variables.
\newblock {\em ZAMM-Zeitschrift fur Angewandte Mathematik und Mechanik}, 2005.

\bibitem{kingma2013auto}
D.~P. Kingma and M.~Welling.
\newblock Auto-encoding variational {B}ayes.
\newblock In {\em Proceedings of the International Conference on Learning
  Representations (ICLR)}, 2014.

\bibitem{leimkuhler2016efficient}
B.~Leimkuhler and C.~Matthews.
\newblock Efficient molecular dynamics using geodesic integration and
  solvent--solute splitting.
\newblock In {\em Proc. R. Soc. A}. The Royal Society, 2016.

\bibitem{leimkuhler2004simulating}
B.~Leimkuhler and S.~Reich.
\newblock {\em Simulating {H}amiltonian dynamics}.
\newblock Cambridge University Press, 2004.

\bibitem{leimkuhler1994symplectic}
B.~J. Leimkuhler and R.~D. Skeel.
\newblock Symplectic numerical integrators in constrained {H}amiltonian
  systems.
\newblock {\em Journal of Computational Physics}, 1994.

\bibitem{lelievre2012langevin}
T.~Leli{\`e}vre, M.~Rousset, and G.~Stoltz.
\newblock Langevin dynamics with constraints and computation of free energy
  differences.
\newblock {\em Mathematics of computation}, 2012.

\bibitem{marin2012approximate}
J.-M. Marin, P.~Pudlo, C.~P. Robert, and R.~J. Ryder.
\newblock Approximate {B}ayesian computational methods.
\newblock {\em Statistics and Computing}, 2012.

\bibitem{marjoram2003markov}
P.~Marjoram, J.~Molitor, V.~Plagnol, and S.~Tavar{\'e}.
\newblock Markov chain {M}onte {C}arlo without likelihoods.
\newblock {\em Proceedings of the National Academy of Sciences}, 2003.

\bibitem{meeds2015hamiltonian}
E.~Meeds, R.~Leenders, and M.~Welling.
\newblock Hamiltonian {ABC}.
\newblock In {\em Proceedings of 31st Conference of Uncertainty in Artificial
  Intelligence}, 2015.

\bibitem{meeds2015optimization}
T.~Meeds and M.~Welling.
\newblock Optimization {M}onte {C}arlo: Efficient and embarrassingly parallel
  likelihood-free inference.
\newblock In {\em Advances in Neural Information Processing Systems}, 2015.

\bibitem{more1980user}
J.~J. Mor\'{e}, B.~S. Garbow, and K.~E. Hillstrom.
\newblock {\em {User Guide for MINPACK-1}}.
\newblock ANL-80-74, Argonne National Laboratory, 1980.

\bibitem{murray2016differentiation}
I.~Murray.
\newblock Differentiation of the {C}holesky decomposition.
\newblock {\em arXiv preprint arXiv:1602.07527}, 2016.

\bibitem{murray2015pseudo}
I.~Murray and M.~M. Graham.
\newblock Pseudo-marginal slice sampling.
\newblock In {\em International Conference on Artificial Intelligence and
  Statistics}, 2016.

\bibitem{neal2011mcmc}
R.~M. Neal.
\newblock {MCMC} using {H}amiltonian dynamics.
\newblock {\em Handbook of Markov Chain Monte Carlo}, 2011.

\bibitem{plummer2006coda}
M.~Plummer, N.~Best, K.~Cowles, and K.~Vines.
\newblock {CODA}: Convergence diagnosis and output analysis for {MCMC}.
\newblock {\em R News}, 2006.

\bibitem{powell1970hybrid}
M.~J.~D. Powell.
\newblock {\em Numerical Methods for Nonlinear Algebraic Equations}, chapter A
  Hybrid Method for Nonlinear Equations.
\newblock Gordon and Breach, 1970.

\bibitem{ratmann2009model}
O.~Ratmann, C.~Andrieu, C.~Wiuf, and S.~Richardson.
\newblock Model criticism based on likelihood-free inference, with an
  application to protein network evolution.
\newblock {\em Proceedings of the National Academy of Sciences}, 2009.

\bibitem{rezende2014stochastic}
D.~J. Rezende, S.~Mohamed, and D.~Wierstra.
\newblock Stochastic backpropagation and approximate inference in deep
  generative models.
\newblock In {\em Proceedings of The 31st International Conference on Machine
  Learning}, 2014.

\bibitem{robert2010model}
C.~P. Robert, K.~Mengersen, and C.~Chen.
\newblock Model choice versus model criticism.
\newblock {\em Proceedings of the National Academy of Sciences of the United
  States of America}, 2010.

\bibitem{ryckaert1977numerical}
J.-P. Ryckaert, G.~Ciccotti, and H.~J. Berendsen.
\newblock Numerical integration of the {C}artesian equations of motion of a
  system with constraints: molecular dynamics of n-alkanes.
\newblock {\em Journal of Computational Physics}, 1977.

\bibitem{schulman2015gradient}
J.~Schulman, N.~Heess, T.~Weber, and P.~Abbeel.
\newblock Gradient estimation using stochastic computation graphs.
\newblock In {\em Advances in Neural Information Processing Systems}, 2015.

\bibitem{theano2016theano}
{Theano Development Team}.
\newblock {Theano: A {Python} framework for fast computation of mathematical
  expressions}.
\newblock {\em arXiv e-prints}, abs/1605.02688, 2016.

\bibitem{welling2011bayesian}
M.~Welling and Y.~W. Teh.
\newblock Bayesian learning via stochastic gradient {L}angevin dynamics.
\newblock In {\em Proceedings of the 28th International Conference on Machine
  Learning}, 2011.

\bibitem{wilkinson2013approximate}
R.~D. Wilkinson.
\newblock Approximate {B}ayesian computation ({ABC}) gives exact results under
  the assumption of model error.
\newblock {\em Statistical applications in genetics and molecular biology},
  2013.

\end{thebibliography}

\end{document}